\documentclass[11pt,letterpaper]{article}

\usepackage[T1]{fontenc}
\usepackage{fullpage}
\usepackage{mathpazo}
\usepackage{amsmath,amssymb,amsthm,amsfonts}
\usepackage{array}
\usepackage{booktabs}
\usepackage{float}
\usepackage[dvipsnames]{xcolor}
\usepackage{microtype}
\usepackage{tabularx}
\usepackage{natbib}
\usepackage{enumitem}
\usepackage{listings}
\usepackage[colorlinks=true,allcolors=magenta]{hyperref}
\usepackage[nameinlink,capitalize,noabbrev]{cleveref}
\usepackage{mathtools}
\usepackage{pifont}
\usepackage{thm-restate}
\usepackage{nicefrac}
\usepackage{tikz}
\usetikzlibrary{arrows.meta}
\usepackage{xcolor}

\usepackage[ruled]{algorithm2e} 
\SetAlFnt{\small}
\SetAlCapFnt{\small}
\SetAlCapNameFnt{\small}
\SetAlCapHSkip{0pt}
\IncMargin{-\parindent}
\crefname{algocf}{Algorithm}{Algorithms}
\Crefname{algocf}{Algorithm}{Algorithms}

\newtheorem{theorem}{Theorem}
\newtheorem{proposition}{Proposition}
\newtheorem{corollary}{Corollary}
\newtheorem{lemma}{Lemma}
\newtheorem{open}{Open Question}
\theoremstyle{definition}
\newtheorem{definition}{Definition}
\newtheorem{example}{Example}
\theoremstyle{remark}
\newtheorem{remark}{Remark}
\newtheorem{observation}{Observation}

\newcommand{\bR}{\mathbb{R}}
\usepackage{tcolorbox}

\renewcommand{\ge}{\geqslant}
\renewcommand{\geq}{\geqslant}
\renewcommand{\le}{\leqslant}
\renewcommand{\leq}{\leqslant}

\newcommand{\supp}{\operatorname{supp}}
\newcommand{\OPT}{\mathcal{O}}

\setcitestyle{numbers,square}

\newcommand{\cA}{\mathcal{A}}
\newcommand{\cF}{\mathcal{F}}
\newcommand{\cW}{\mathcal{W}}
\newcommand{\cI}{\mathcal{I}}
\newcommand{\cK}{\mathcal{K}}
\newcommand{\cV}{\mathcal{V}}
\newcommand{\cVz}{\mathcal{V}_0}
\newcommand{\cVo}{\mathcal{V}_1}
\newcommand{\lowerstar}{\mathop{\underline{\partial}}}

\newcommand{\EFone}{\mathrm{EF1}^{+}_{-}}
\newcommand{\EFoneplus}{\mathrm{EF1}^{+}}
\newcommand{\EFoneminus}{\mathrm{EF1}_{-}}
\newcommand{\EFXpm}{\mathrm{EFX}^{+}_{-}}
\newcommand{\EFXzm}{\mathrm{EFX}^{0}_{-}}
\newcommand{\EFXpz}{\mathrm{EFX}^{+}_{0}}
\newcommand{\EFXzz}{\mathrm{EFX}^{0}_{0}}

\newcommand{\ind}[1]{\mathbf{1}\!\left[#1\right]}
\newcommand{\Ap}{A^{+}}
\newcommand{\Am}{A^{-}}

\hypersetup{pdftitle={Fair Division under Unnormalized Boolean Valuations},pdfauthor={}}

\title{Fair Division Under Boolean Valuations:\\Beyond Normalization}

\author{Nisarg Shah\\University of Toronto\\\texttt{nisarg@cs.toronto.edu} \and Paritosh Verma\\University of Toronto\\\texttt{paritosh.verma@utoronto.ca}}
\makeatletter
\let\@date\@empty
\makeatother

\begin{document}

\maketitle

\begin{abstract}
We study fair division of indivisible items when agents have arbitrary two-level preferences: the value of each agent for any set of items is Boolean, which need not be monotone or additive. Notably, we do not impose the standard assumption of \emph{normalization}, i.e., different agents may value the empty set at different Boolean levels. Since the preferences are nonmonotone, \emph{envy-freeness up to one item} (EF1) and \emph{envy-freeness up to any item} (EFX) each admit several variants, depending on which items are tested for removal and whether they are removed from the envious agent's bundle or the envied agent's bundle. This paper investigates the existence of these variants of EF1 and EFX, on their own and together with economic efficiency, incentive compatibility, feasibility constraints, and lottery-based randomization. Our paper presents a wide range of positive and negative results that collectively highlight an interesting phenomenon: the landscape of achievable guarantees depends crucially on the number of normalized agents, who value the empty bundle at the lower Boolean level.

\end{abstract}

\newpage
\setcounter{tocdepth}{2}
\tableofcontents
\medskip
\newpage

\section{Introduction}
 
Fair division studies how to allocate scarce resources --- indivisible items, chores, time slots, inheritance, course seats --- among agents with heterogeneous preferences, so that no agent feels that they were treated unfairly. Much of this theory, however, assumes that agents' preferences are \emph{monotone} and \emph{normalized}. A valuation is monotone nondecreasing if adding items always weakly increases utility, and monotone nonincreasing if adding items always weakly decreases utility. Normalization, which is an even more prevalent assumption, states that an agent's value for the empty set of items is zero. These assumptions primarily seem to serve the purpose of mathematical convenience. As one would expect, dropping monotonicity brings considerable nuance and complexity: well-known results for monotone nondecreasing valuations often fail outright or demand substantially new arguments for nonmonotone valuations~\cite{barman2026fair,bilo2026approximately}. In a similar vein, normalization plays a central role in enabling fundamental fair-division results, such as Su's existence theorem for envy-free cake divisions~\cite{edward1999rental}.

However, in many real-world settings, agents' preferences obey neither monotonicity nor normalization. Having more of something is not always better: while taking two or three courses might benefit a student, taking a dozen courses can become burdensome and counterproductive. Similarly, advising a few students might make a researcher more productive overall, but can become a chore if the number of advisees increases beyond a threshold. Free disposal of items, which leads to monotone nondecreasing preferences,\footnote{If items can be disposed of, then any new item that leads to a decrease in utility can be disposed of at the agent's end, making preferences monotone nondecreasing.} is also unavailable in settings such as chore allocation. It is likewise natural for agents to assign a nonzero value to the empty bundle: a student might prefer taking no courses in a semester to taking an uninteresting course; if the items are perceived as chores, the empty bundle can be more valuable than nonempty bundles.

The existence of desirable allocations without any assumptions like normalization and monotonicity is therefore just as fundamental a question as its counterpart, yet this question remains underexplored. To isolate this question in its cleanest form, we study fair division of indivisible items under \emph{Boolean valuations}: every agent \(i\) has an arbitrary set valuation \(v_i:2^M\to\{0,1\}\), with no monotonicity, additivity, or normalization assumption. This is the most general two-level preference over bundles one can model using set functions. At the ordinal level, it is simply a preference with two indifference classes, and the numerical values we assign are immaterial. Previous work~\cite{berczi2024,bhaskar2025} has considered important subclasses of Boolean valuations in which every agent assigns the same value to the empty bundle, but not arbitrary mixed profiles containing agents of both types. Additionally, prior work focused primarily on fairness as the main objective; the interplay and tradeoffs between fairness and other desiderata remained underexplored.

The goal of this paper is to analyze the fairness guarantees that are possible for general Boolean valuations, with and without additional desiderata. Along with fairness, we consider economic efficiency, incentive compatibility for individuals and coalitions, structured feasibility constraints, and randomization over allocations.

As our results highlight, a single parameter turns out to govern much of the resulting landscape of attainable guarantees: the number of agents that place the empty bundle in the lower Boolean level. To intuitively understand the relevance of this parameter, consider an instance with two agents where the first agent values the empty bundle at zero, while the second agent values the empty bundle at one. Irrespective of their valuations for nonempty bundles, giving all the items to the first agent results in an envy-free allocation. In the later sections we will see, that his phenomenon carries over and governs the existence of fine-grained fairness guarantees, with and without, other desiderata.

We will refer to the setting where \(v_i(\varnothing)=0\) for all agents, as the \emph{normalized setting}. If we have \(v_i(\varnothing)=1\) for everyone, then a simple linear translation produces an equivalent valuation having range \(\{-1,0\}\) where the empty bundle is valued at zero; following B\'erczi et al.~\cite{berczi2024}, we will refer to this as the \emph{negative-Boolean} setting. The setting we study allows different agents to place the empty bundle at possibly different Boolean levels, and hence it generalizes both these settings.

\subsection*{Our contributions}
 
To the best of our knowledge, our paper is the first to consider fair division of indivisible items beyond the normalization assumption; dropping this assumption enables the model to capture the full subjectivity of agents' preferences, even regarding the value of the empty bundle. As stated above, equivalently, our paper can be seen as studying fair division under the most general two-level preferences.

Our contributions span five related directions: we chart the landscape of existence of envy-based fairness guarantees, resolve an open problem of B\'erczi et al.~\cite{berczi2024}, and establish positive and negative results on the interplay between fairness and other desiderata. In detail, our contributions are the following.
 
\medskip\textbf{The EF1 landscape} (\Cref{sec:existence,sec:onesided}). We begin by studying the variants of EF1 for nonmonotone valuations that are based on how the envy gets eliminated: removal of an item from the envied bundle (\(\EFoneplus\)), from the envious agent's own bundle (\(\EFoneminus\)), or from either side (\(\EFone\)); formal definitions appear in \Cref{sec:efx-notions}. Our starting point is a simple observation combining two known results: every Boolean profile, with no assumption on normalization, admits a complete \(\EFone\) allocation (\Cref{thm:universal-ef1}). We then determine the tight conditions under which each stronger variant of EF1 exists, on its own and with Pareto optimality. Let \(Z\) denote the set of agents that value the empty bundle at zero. The value of $|Z|$ governs the existence of the stronger variants. Allocations satisfying \(\EFoneplus\) are guaranteed to exist if and only if \(|Z|\geq 1\), i.e., as long as there is at least one agent that values the empty bundle at zero. On the other hand, allocations satisfying \(\EFoneminus\) exist if and only if \(|Z|\le1\). For $|Z| = 1$, we can, in fact, always guarantee exact envy-freeness. Conditions characterizing the existence of fair and efficient allocations also depend on $|Z|$: \(\EFoneplus\) and PO allocations exist if and only if $|Z| \geq 1$, and  \(\EFoneminus\) and PO allocations exist if and only if $|Z| =1$. These characterizations appear in \Cref{ithm:trich}. We find it rather surprising that the characterizing conditions depend on a single parameter, \(|Z|\), rather than on items being goods or chores.

\textbf{The EFX landscape and an open question} (\Cref{sec:efx}). We examine the existence of the four variants of EFX --- $\EFXpm, \EFXzm, \EFXpz,$ and $\EFXzz$ --- defined by B\'erczi et al.~\cite{berczi2024}. Their definitions appear in \Cref{sec:efx-notions}, and \Cref{fig:implications} maps their relationships with the three variants of EF1. The weakest notion, \(\EFXpm\), coincides with EF1 on the Boolean domain and hence always exists (\Cref{prop:efx-ef1}). \(\EFXzm\) is guaranteed whenever \(|Z| \geq 1\), and this characterizing condition is tight because a profile with \(|Z| = 0\) may fail to admit \(\EFXzm\) allocations (\Cref{ilem:plus,thm:efx0minus}). The strongest notion, \(\EFXzz\), admits no universal guarantee~\cite{berczi2024}. For the fourth notion, \(\EFXpz\), we resolve an open problem, Question~16, posed by B\'erczi et al.~\cite{berczi2024}: every instance with negative-Boolean valuations admits a complete \(\EFXpz\) allocation (\Cref{thm:negative-efx}), where the existence was previously known only for identical valuations~\cite[Theorem~9]{berczi2024}. The proof, which is the technical highlight of the work, is based on a novel peel-and-match algorithm, in which candidate bundles are assigned to agents by a matching via Hall's theorem and, crucially, a failure of Hall's condition forms the basis of an induction argument rather than a dead end.

\textbf{Fair, efficient, and weakly group-strategyproof mechanisms} (\Cref{sec:mechanisms}). We strengthen B\'erczi et al.'s result on the existence of \(\EFXzm\) allocations~\cite[Theorem~7]{berczi2024} for the normalized setting ($|Z| = n$). We give a simple weakly group-strategyproof mechanism that always outputs an \(\EFXzm\) and PO allocation. Given the reported valuations, the mechanism first maximizes social welfare (the sum of the values that the agents get in an allocation), subject to that, then minimizes the total number of items allocated to agents that receive a value of one, and finally uses a fixed tie-breaking rule. 

\textbf{Best-of-both-worlds guarantees} (\Cref{sec:bobw}). We show that if $|Z| \geq 1$, i.e., at least one agent values the empty set at zero, then there exists a lottery over allocations that is ex-ante envy-free, and every allocation in its support is Pareto optimal, \(\EFoneplus\), and \(\EFXzm\) (\Cref{thm:bobwZ}).\footnote{This can be viewed as another extension of the existence of \(\EFXzm\) allocations~\cite[Theorem~7]{berczi2024}.}
 
\textbf{Matroid constraints} (\Cref{sec:matroid}). We consider a fair division setting where allocated bundles must be independent sets of a matroid. We characterize, for two agents, exactly which matroids guarantee a feasible EF1 allocation for all Boolean profiles (\Cref{thm:matroid-characterization}). Notably, the characterizing condition (existence of a self-complementary component of the symmetric basis-pair exchange graph) connects the problem to the well-known reconfiguration conjectures of White~\cite{white1980} and Gabow~\cite{gabow1976} in matroid theory. Known reconfiguration theorems yield the existence of EF1 allocations for strongly base-orderable, split, sparse paving, and regular matroids (\Cref{cor:matroid-classes}). A partition matroid, by contrast, makes EF1 and Pareto optimality incompatible (\Cref{prop:matroid-po}).

There has been extensive work on the existence of EF1 allocations for various valuation classes. Moreover, the interplay of fairness, efficiency, and incentives has also received considerable attention for related valuation classes like dichotomous valuations. %See \Cref{app:related-work}, for a discussion on related work.

\section{Related work}
\label{app:related-work}

We now mention some relevant work from the literature.

\paragraph{Existence of EF1.} The existence of EF1 allocations and its variants is known for various valuation classes; \Cref{tab:known} lists the known results and \Cref{sec:valuation-classes} lists the definitions of the concerned valuation classes. Envy-cycle elimination gives goods-side EF1 for monotone nondecreasing valuations~\cite{lipton2004}. EF1 also exists for additive mixed manna~\cite{aziz2022}, for doubly monotone valuations~\cite{bhaskar2021}, and for two agents with arbitrary set valuations~\cite{berczi2024}. Most directly relevant here, B\'erczi et al.\ \cite{berczi2024} prove \(\EFXzm\), and hence EF1, for normalized Boolean valuations~\cite[Theorem~7]{berczi2024}. They also prove the stronger \(\EFXpz\) guarantee for identical negative-Boolean valuations~\cite[Theorem~9]{berczi2024} and pose an open problem of whether the existence result holds for nonidentical negative-Boolean valuations~\cite[Question~16]{berczi2024}. Bhaskar et al.\ subsequently prove EF1 for non-identical negative-Boolean valuations and an arbitrary number of agents~\cite[Theorem~3]{bhaskar2025}, but the EFX question remained unresolved; we resolve it here. 

\begin{table}[t]
\centering
\small
\renewcommand{\arraystretch}{1.2}
\begin{tabularx}{\columnwidth}{@{}>{\raggedright\arraybackslash}p{0.2\textwidth}>{\raggedright\arraybackslash}p{0.05\textwidth}X@{}}
\toprule
Domain & Guarantee & Source\\
\midrule
Arbitrary valuations & \(\EFone\) & Open\\
Nonnegative (nonmonotone) valuations & \(\EFone\) & Open\\
Monotone nondecreasing set valuations & \(\EFoneplus\) & Lipton et al.~\cite{lipton2004}\\
Additive mixed manna & \(\EFone\) & Aziz et al.~\cite{aziz2022}\\
Doubly monotone valuations & \(\EFone\) & Bhaskar et al.~\cite{bhaskar2021}\\
Two agents, arbitrary valuations & \(\EFone\) & B\'erczi et al.~\cite{berczi2024}\\
Normalized Boolean valuations& \(\EFone\) & B\'erczi et al.~\cite[Theorem~7]{berczi2024}\\
Negative Boolean valuations & \(\EFone\) & Bhaskar et al.~\cite[Theorem~3]{bhaskar2025}\\
\bottomrule
\end{tabularx}
\caption{EF1 existence results for different valuation classes.}
\label{tab:known}
\end{table}

\paragraph{Dichotomous valuations.}
A large body of work concerns valuations whose \emph{marginals} are binary, rather than whose values are binary. For binary additive valuations, where every item is worth zero or one and values add up, Halpern et al.~\cite{halpern2020binary} show that maximum Nash welfare with a fixed tie-breaking rule is simultaneously weakly group-strategyproof, EF1, and Pareto optimal. Babaioff et al.~\cite{babaioff2021dichotomous} extend the positive result to submodular valuations with binary marginals, also known as matroid-rank valuations: their deterministic truthful mechanism returns a Lorenz-dominating allocation, which is in particular EFX and of maximum Nash welfare, and randomizing over agent priorities makes it envy-free ex ante while retaining these properties ex post. Barman and Verma~\cite{barman-verma-mms} prove that allocations that are both maximin fair and Pareto optimal exist for matroid-rank valuations, and subsequently~\cite{barman-verma-truthful} that the mechanism of Babaioff et al.\ is in fact weakly group-strategyproof for EF1. These papers call the property group strategyproofness and explicitly distinguish it from strong group strategyproofness. Viswanathan and Zick~\cite{viswanathan-zick-yankee} give simpler and faster mechanisms that are strategyproof and output welfare-maximizing allocations for a wide range of fairness objectives~\cite{viswanathan-zick-framework}; see also Benabbou et al.~\cite{benabbou2021matroid} for fairness and efficiency in this class.

These classes and ours share the feature of being characterized by two levels, but they are not comparable. A valuation with binary marginals is monotone nondecreasing and its range is \(\{0,1,\ldots,m\}\), whereas a Boolean valuation has range \(\{0,1\}\) and is subject to no monotonicity, submodularity, or additivity requirement.

\section{Preliminaries}
\label{sec:preliminaries}
  
Let \(N=\{1,\ldots,n\}\) be a finite set of agents and \(M=\{1,\ldots,m\}\) a finite set of indivisible items. Subsets \(S \subseteq M\) are called \emph{bundles}. A complete allocation \(A=(A_1,\ldots,A_n)\) is a partition of \(M\): \(\bigcup_{i\in N}A_i=M\) and \(A_i\cap A_j=\varnothing\) for all distinct \(i,j\in N\). Each agent \(i\) has an arbitrary valuation \(v_i:2^M\to\{0,1\}\). Let \(\cV\) denote this domain, \(\cVz=\{v\in\cV:v(\varnothing)=0\}\), and \(\cVo=\{v\in\cV:v(\varnothing)=1\}\). For a given instance, we define \(Z=\{i\in N:v_i(\varnothing)=0\}\) and \(E=N\setminus Z\) as the sets of agents that value the empty bundle at zero and at one, respectively.

An allocation \(B\) \emph{Pareto dominates} an allocation \(A\) if \(v_i(B_i)\ge v_i(A_i)\) for every agent \(i\) and the inequality is strict for at least one agent. An allocation is \emph{Pareto optimal} (PO) if no allocation Pareto dominates it.
 
Only the order of the two values matters for every fairness notion considered here.\footnote{The numerical values become relevant only for cardinal properties, such as proportionality or welfare approximation.} Consequently, all results extend verbatim to arbitrary preferences with two indifference classes by mapping the lower class to zero and the higher class to one. The following perspective is often useful: for \(v_i\in\cV\), define the normalized translation \(u_i(S)=v_i(S)-v_i(\varnothing)\). If \(v_i\in\cVz\), then \(u_i\) has range \(\{0,1\}\); if \(v_i\in\cVo\), then \(u_i\) has range \(\{-1,0\}\). Since the transformation is an agent-specific additive translation, it preserves every comparison between two bundles and hence the set of fair (EF1 or EFX) allocations remains unchanged.

\subsection{Three versions of EF1}
We begin by defining variants of envy-freeness up to one item (EF1) that we focus on. Note that, all these notions restore envy-freeness via a removal of an item from some agent's bundle.
\begin{definition}[EF1 variants]
We focus on the following three variants of EF1,
\begin{itemize}
    \item  An allocation \(A\) satisfies \emph{goods-side envy-freeness up to one item} (\(\EFoneplus\)) if for all pairs of agents \(i,j \in N\), we have \(v_i(A_i)\ge v_i(A_j \setminus S)\), for some subset $S \subseteq A_j$ satisfying $|S| \leq 1$.
    \item  An allocation \(A\) satisfies \emph{chore-side envy-freeness up to one item} (\(\EFoneminus\)) if for all pairs of agents \(i,j \in N\), we have \(v_i(A_i \setminus S)\ge v_i(A_j)\) for some subset $S \subseteq A_i$ satisfying $|S| \leq 1$. 
    \item An allocation $A$ satisfies \emph{two-sided envy-freeness up to one item} (\(\EFone\)) if either removal is allowed, i.e., for every \(i,j\), we have \(v_i(A_i\setminus S)\ge v_i(A_j\setminus S)\) for some subset $S \subseteq A_i \cup A_j$ satisfying $|S| \leq 1$.
\end{itemize}

\end{definition}
 
The superscript \(+\) records removal from the envied bundle and the subscript \(-\) removal from the envious agent's bundle, matching the EFX notation of B\'erczi et al. Both \(\EFoneplus\) and \(\EFoneminus\) imply \(\EFone\). We use ``EF1'' without decoration for \(\EFone\), the removal-from-either-side notion standard for mixed manna.
 
For a bundle \(S\subseteq M\), define its one-deletion lower star by
\[
\lowerstar S=\{S\}\cup\{S\setminus\{e\}:e\in S\}.
\]
For a Boolean valuation \(v\), call \(S\) \emph{unsafe} if \(v(T)=0\) for every \(T\in\lowerstar S\), and \emph{robust} if \(v(T)=1\) for every \(T\in\lowerstar S\). In particular, \(\varnothing\) is unsafe for every \(v\in\cVz\) and robust for every \(v\in\cVo\).
 
\begin{proposition}[Violation characterization]
\label{lem:violation}
For every Boolean profile, allocation \(A\), and ordered pair \(i,j\), the pair \(i,j\) violates EF1 if and only if \(A_i\) is unsafe for \(v_i\) and \(A_j\) is robust for \(v_i\).
\end{proposition}
 
\begin{proof}
Under Boolean valuations, envy means \(v_i(A_i)=0<1=v_i(A_j)\). Removing an item from \(A_i\) fails to eliminate this envy for every possible removal exactly when \(v_i(A_i\setminus\{e\})=0\) for every \(e\in A_i\). Removing an item from \(A_j\) fails for every possible removal exactly when \(v_i(A_j\setminus\{e\})=1\) for every \(e\in A_j\). These conditions say precisely that \(A_i\) is unsafe and \(A_j\) is robust.
\end{proof}
 
\subsection{Four nonmonotone EFX relaxations}
\label{sec:efx-notions}
 
B\'erczi et al.~\cite{berczi2024} define four variants of envy-freeness up to any item (EFX) based on which items must be removed to restore envy-freeness. The following definitions will be useful for defining these notions. For an agent \(i\) and a bundle \(X \subseteq M\) define,
\[
\begin{aligned}
S_i^+(X)&=\{e\in X:v_i(X)-v_i(X\setminus\{e\})>0\},\\
S_i^-(X)&=\{e\in X:v_i(X)-v_i(X\setminus\{e\})<0\},\\
S_i^0(X)&=\{e\in X:v_i(X)-v_i(X\setminus\{e\})=0\}.
\end{aligned}
\]

We will describe the notions of EFX defined in B\'erczi et al. using \Cref{tab:efx-notions}, where each row specifies a variant of EFX. An allocation $A$ satisfies a particular variant of EFX (corresponding to a row) if for all pairs of agents $i,j \in N$, either we have $v_i(A_i) \geq v_i(A_j)$, or two conditions hold: $(a)$ we have $v_i(A_i \setminus \{e\}) \geq v_i(A_j \setminus \{e\})$ \emph{for all} $e \in T_i \cup T_j$ where $T_i \subseteq A_i$ and $T_j \subseteq A_j$ are the items we test from the envious and the envied agent's bundles respectively, and $(b)$ $T_i \cup T_j$ is non-empty. Columns of \Cref{tab:efx-notions} list the sets of items $T_i$ and $T_j$ that are tested for each of the four notions of EFX.

\begin{table}%[H]
\centering
\small
\renewcommand{\arraystretch}{1.2}
\begin{tabular}{@{}ccc@{}}
\toprule
Notion & $T_j$ (envied agent's bundle) & $T_i$ (envious agent's bundle)\\
\midrule
\(\EFXzz\) & \(S_i^+(A_j)\cup S_i^0(A_j)\) & \(S_i^-(A_i)\cup S_i^0(A_i)\)\\
\(\EFXzm\) & \(S_i^+(A_j)\cup S_i^0(A_j)\) & \(S_i^-(A_i)\)\\
\(\EFXpz\) & \(S_i^+(A_j)\) & \(S_i^-(A_i)\cup S_i^0(A_i)\)\\
\(\EFXpm\) & \(S_i^+(A_j)\) & \(S_i^-(A_i)\)\\
\bottomrule
\end{tabular}
\caption{The four nonmonotone EFX relaxations of B\'erczi et al.~\cite{berczi2024}: for an envious ordered pair \(i,j\), the items tested in the envied bundle \(A_j\) and in the envious agent's own bundle \(A_i\).}
\label{tab:efx-notions}
\end{table}
 
Since the tested set of every notion must be nonempty and each tested item eliminates the envy from its side, each of the four notions implies \(\EFone\). The four notions are moreover partially ordered, as in \Cref{fig:implications}: for an envious pair, a tested item that does not change the bundle's value can never eliminate the envy from its side, so a notion that tests an \(S^0\)-set forces that set to be empty, and the tested sets of the stronger and the weaker notion then coincide. Hence \(\EFXzz\) implies both \(\EFXzm\) and \(\EFXpz\), each of which implies \(\EFXpm\). Two further implications hold under normalization: on \(\cVz^n\) no agent envies an empty bundle, so \(\EFXzm\) implies \(\EFoneplus\); and on \(\cVo^n\) an envious agent never holds an empty bundle, so \(\EFXpz\) implies \(\EFoneminus\). Neither dashed implication holds in general: a one-item instance with an empty \emph{envied} bundle (respectively, an empty \emph{envious} bundle) separates the notions. \Cref{fig:implications} collects the seven fairness notions of this section, together with the existence results that this paper establishes for each of them for Boolean valuations.
 
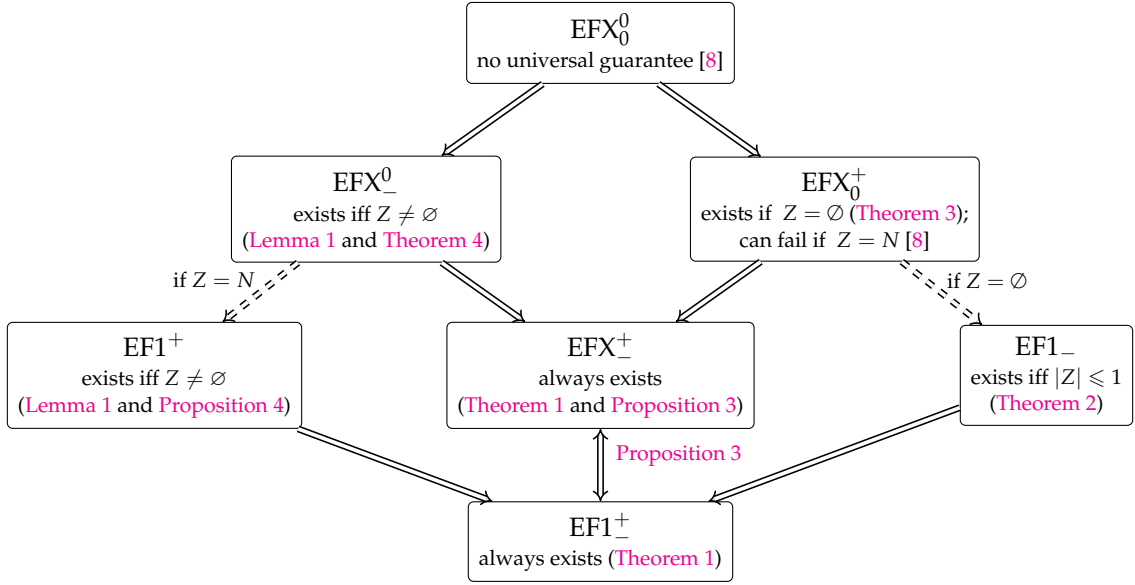
\begin{figure}[htb!]
\centering
\begin{tikzpicture}[
  imp/.style={-{Implies},double distance=1.4pt,semithick},
  cimp/.style={-{Implies},double distance=1.4pt,semithick,dashed},
  equ/.style={{Implies}-{Implies},double distance=1.4pt,semithick},
  notion/.style={draw,rounded corners=2pt,align=center,inner sep=4pt,font=\small},
  scale=1
]

% Compressed vertical coordinates
\node[notion] (zz) at (0, 7.5)
  {\(\EFXzz\)\\[-1pt]{\scriptsize no universal guarantee~\cite{berczi2024}}};
  
\node[notion] (zm) at (-2.8, 5.5)
  {\(\EFXzm\)\\[-1pt]{\scriptsize exists iff \(Z\ne\varnothing\)}\\[-2pt]{\scriptsize(\Cref{ilem:plus,thm:efx0minus})}};
  
\node[notion] (pz) at (2.8, 5.5)
  {\(\EFXpz\)\\[-1pt]{\scriptsize exists if \(\ Z = \emptyset\)}\\[-2pt]{\scriptsize (\Cref{thm:negative-efx});}\\[-2pt]{\scriptsize can fail if \(\ Z = N \)~\cite{berczi2024}}};
  
\node[notion] (efp) at (-4.2, 2.5)
  {\(\EFoneplus\)\\[-1pt]{\scriptsize exists iff \(Z\ne\varnothing\)}\\[-2pt]{\scriptsize(\Cref{ilem:plus,iprop:noplus})}};
  
\node[notion] (pm) at (0, 2.5)
  {\(\EFXpm\)\\[-1pt]{\scriptsize always exists}\\[-2pt]{\scriptsize(\Cref{thm:universal-ef1,prop:efx-ef1})}};
  
\node[notion] (efm) at (4.2, 2.5)
  {\(\EFoneminus\)\\[-1pt]{\scriptsize exists iff \(|Z|\le1\)}\\[-2pt]{\scriptsize(\Cref{ithm:trich})}};
  
\node[notion] (ef) at (0, 0)
  {\(\EFone\)\\[-1pt]{\scriptsize always exists (\Cref{thm:universal-ef1})}};

% Edges
\draw[imp] (zz) -- (zm);
\draw[imp] (zz) -- (pz);
\draw[imp] (zm) -- (pm);
\draw[imp] (pz) -- (pm);
\draw[cimp] (zm) -- (efp) node[midway,above left=-2pt,font=\scriptsize]{if \(Z = N\)};
\draw[cimp] (pz) -- (efm) node[midway,above right=-2pt,font=\scriptsize]{if \( Z = \emptyset \)};
\draw[imp] (efp) -- (ef);
\draw[imp] (efm) -- (ef);
\draw[equ] (pm) -- (ef) node[midway,right=2pt,font=\scriptsize,align=left]{\Cref{prop:efx-ef1}\\};

\end{tikzpicture}
\caption{Implications among the three EF1 versions and the four EFX relaxations. Solid double arrows hold for arbitrary set valuations; the equivalence \(\EFXpm\Leftrightarrow\EFone\) is specific to the Boolean domain (\Cref{prop:efx-ef1}); and the dashed double arrows hold only on the indicated normalization endpoints, where the relevant empty bundle cannot occur. Each vertex records the existence status of its notion on arbitrary Boolean profiles, with the results establishing it.}
\label{fig:implications}
\end{figure}

\section{EF1 Existence for Unnormalized Boolean Valuations}
\label{sec:existence}
 
We begin by settling the existence of EF1 allocations. The following theorem is a short synthesis of two established results. It nevertheless serves as the starting point for the paper: in \Cref{sec:onesided,sec:efx} we prove stronger, fine-grained existence guarantees, and \Cref{sec:mechanisms,sec:bobw} combine fairness with economic efficiency and incentive compatibility.
 
\begin{restatable}{theorem}{efoneExistence}[EF1 existence]
\label{thm:universal-ef1}
Every instance with arbitrary valuations \(v_i:2^M\to\{0,1\}\) admits a complete EF1 allocation. \end{restatable}
\begin{proof}
Let \(E=N\setminus Z\) be the agents that value the empty bundle at one. If \(Z\ne\varnothing\), give the empty bundle to every agent in \(E\), and allocate all items among the agents in \(Z\) using the \(\EFXzm\) existence result of B\'erczi et al.\ for heterogeneous normalized positive-Boolean valuations~\cite[Theorem~7]{berczi2024}. This gives the required guarantee within \(Z\). Every \(i\in E\) receives value \(v_i(\varnothing)=1\), so \(i\) envies no agent, and every \(i\in Z\) assigns value \(v_i(\varnothing)=0\) to the bundle of each agent in \(E\), which is no greater than \(v_i(A_i)\). Hence every comparison involving \(E\) is envy-free, and the full allocation is \(\EFXzm\), in particular EF1. If instead \(Z=\varnothing\), then the transformation \(u_i(S)=v_i(S)-1\) results in a heterogeneous negative-Boolean instance, which admits a complete EF1 allocation by Bhaskar et al.~\cite[Theorem~3]{bhaskar2025}; translation by a constant preserves every inequality in the definition of EF1, so the same allocation is EF1 for \(v\).
\end{proof}

\section{One-Sided EF1: A Fine-Grained Existence Analysis}
\label{sec:onesided}
 
\Cref{sec:existence} settles the existence of the two-sided notion, $\EFone$. The
one-sided notions behave very differently, with the size of the set $Z=\{i:v_i(\varnothing)=0\}$
governing their existence exactly. We show that the guarantees for $\EFoneplus$ and $\EFoneminus$, both alone and
in conjunction with Pareto optimality, are controlled by
whether $|Z|$ is $0$, $1$, or at least $2$. The resulting trichotomy is summarized in
\Cref{ithm:trich}; the middle column
$|Z|=1$ is the unique regime in which a single allocation is simultaneously $\EFoneplus$,
$\EFoneminus$, and Pareto optimal.

\begin{restatable}{theorem}{fineGrainedOneSided}\label{ithm:trich}
For Boolean profiles, universal existence of the one-sided notions is governed by $|Z|$:
\[
\renewcommand{\arraystretch}{1.25}
\begin{array}{lccc}
 & |Z|=0\ & |Z|=1 & |Z|\ge 2\\ \hline
\EFoneplus & \text{\ding{55}} & \checkmark & \checkmark\\
\EFoneminus & \checkmark & \checkmark & \text{\ding{55}}\\
\textup{PO}\wedge\EFoneplus & \text{\ding{55}} & \checkmark & \checkmark\\
\textup{PO}\wedge\EFoneminus & \text{\ding{55}} & \checkmark & \text{\ding{55}}
\end{array}
\]
where $\checkmark$ means every instance in that column admits such an allocation and \ding{55}\ means
some instance admits none. Equivalently: $\EFoneplus$ (with or without PO) is guaranteed iff
$|Z| \geq 1$; $\EFoneminus$ is guaranteed iff $|Z|\le 1$; and a Pareto-optimal $\EFoneminus$
allocation is guaranteed iff $|Z|=1$.
\end{restatable}

The middle column of the table in \Cref{ithm:trich} is a sweet spot. When exactly one agent values the empty bundle at zero, the allocation of
\Cref{ilem:minus-Z1} is Pareto optimal, $\EFoneplus$, and $\EFoneminus$. The two one-sided guarantees therefore
overlap only at $|Z|=1$, while $\EFone$ (\Cref{sec:existence}) is the fairness notion that exists across all three columns.

This section is devoted to proving \Cref{ithm:trich}. We begin by defining a useful piece of notation: for every agent $i \in N$ and bundle $S \subseteq M$, let
\[
  v_i'(S)\;=\;\max_{T \in \lowerstar S} v_i(T),
\]
so that $v_i'(S)=0$ exactly when $S$ is unsafe for $v_i$, and $v_i'(S)=1$ otherwise. For agents in
$E=N\setminus Z$ (those with $v_i(\varnothing)=1$) every bundle $S$ of size at most one has
$v_i'(S) =1$, since its lower star contains $\varnothing$. We begin by characterizing one-sided EF1 violations; the following proposition immediately follows from the relevant definitions for Boolean valuations.
 
\begin{proposition}[One-sided violations]\label{ilem:viol}
Fix an instance with Boolean preferences, an allocation $A$, and an ordered pair $i,j$ with $v_i(A_i)=0$ and
$v_i(A_j)=1$.
\begin{enumerate}[label=\textup{(\alph*)}]
  \item The pair $i,j$ violates \textup{EF1}${}^{+}$ iff $A_j$ is robust for $v_i$.
  \item The pair $i,j$ violates \textup{EF1}${}^{-}$ iff $A_i$ is unsafe for $v_i$, i.e.\ $v_i'(A_i)=0$.
\end{enumerate}
\end{proposition}

\subsection{Two score-based allocation rules}
\label{sec:optimizers}
 
We design two algorithms, one for each one-sided notion. Both are \emph{score-based}: each of them maximizes a primary score over all complete allocations and then, among the allocations attaining that maximum, optimizes a secondary score. For any profile of Boolean valuation functions $p=(p_1,\ldots,p_n) \in (\cVz \cup \cVo)^n$ and any allocation $A$, define the four scores
\[
  W_p(A)=\sum_{i} p_i(A_i),\quad C_p(A)=\sum_i p_i(A_i)\,|A_i|,\]\[
  W'_p(A)=\sum_i p_i'(A_i),\quad C'_p(A)=\sum_i p_i'(A_i)\,|A_i|,
\]
where $p_i'(S)=\max_{T\in\lowerstar S}p_i(T)$ as above. When the profile under consideration $p$ is unambiguous, we drop the subscript and simply write $W$, $C$, $W'$, and $C'$. Let $\cA$ denote the finite set of complete allocations.
 
\begin{algorithm}%[H]
\caption{\textsc{GoodsOpt}: the goods-side rule}
\label{alg:goodsopt}
\KwIn{A Boolean profile $p$.}
\KwOut{A complete allocation.}
$\cW\leftarrow\arg\max_{A\in\cA} W_p(A)$\;
$\OPT^{+}(p)\leftarrow\arg\min_{A\in\cW} C_p(A)$\;
\Return{an arbitrary allocation in $\OPT^{+}(p)$}
\end{algorithm}
\begin{algorithm}%[H]
\caption{\textsc{ChoresOpt}: the chores-side rule}
\label{alg:choresopt}
\KwIn{A Boolean profile $p$.}
\KwOut{A complete allocation.}
$\cW\leftarrow\arg\max_{A\in\cA} W'_p(A)$\;
$\OPT^{-}(p)\leftarrow\arg\max_{A\in\cW} C'_p(A)$\;
\Return{an arbitrary allocation in $\OPT^{-}(p)$}
\end{algorithm}
 
Thus \textsc{GoodsOpt} (\Cref{alg:goodsopt}) maximizes $W_p$ and, subject to that, \emph{minimizes} $C_p$; \textsc{ChoresOpt} (\Cref{alg:choresopt}) maximizes $W'_p$ and, subject to that, \emph{maximizes} $C'_p$. We write $\Ap$ and $\Am$ for their outputs on $v$. All fairness and efficiency statements below hold for \emph{every} allocation in the optimal sets $\OPT^{+}(v)$ and $\OPT^{-}(v)$, so the final arbitrary choice of the output allocation is immaterial.\footnote{Later, in \Cref{sec:mechanisms}, we show that \textsc{GoodsOpt} additionally underlies a weakly group-strategyproof and Pareto-optimal mechanism, and in \Cref{sec:bobw} that the uniform lottery over $\OPT^{+}(v)$ is a best-of-both-worlds solution.} Both algorithms optimize over all of $\cA$: they serve the purpose of establishing existence and do not lead to polynomial-time algorithms.
 
Before analyzing these two rules, we show an equivalence that upgrades every EF1 guarantee to the weakest variant of EFX for free: on the Boolean domain, the weakest of the four EFX relaxations of \Cref{sec:efx-notions} is not a strengthening of EF1 at all.
 
\begin{restatable}{proposition}{propEFxEFOneequivalence}
\label{prop:efx-ef1}
For Boolean valuations, an allocation is \(\EFXpm\) if and only if it is EF1.
\end{restatable}
\begin{proof}
Suppose first that \(A\) is \(\EFXpm\). For an envious pair \(i,j\), the set \(S_i^+(A_j)\cup S_i^-(A_i)\) is nonempty, and removal of any item in this set eliminates the envy from the corresponding side. Such an item is an EF1 witness.
 
Conversely, suppose \(A\) is EF1 and \(i\) envies \(j\). Boolean valuations force \(v_i(A_i)=0\) and \(v_i(A_j)=1\). If an EF1 witness \(e\) lies in \(A_j\), then \(v_i(A_j\setminus\{e\})=0\), so \(e\in S_i^+(A_j)\). If it lies in \(A_i\), then \(v_i(A_i\setminus\{e\})=1\), so \(e\in S_i^-(A_i)\). Thus the tested set for \(\EFXpm\) is nonempty. Moreover, every item in \(S_i^+(A_j)\) changes the Boolean value of \(A_j\) from one to zero upon deletion, and every item in \(S_i^-(A_i)\) changes the value of \(A_i\) from zero to one upon deletion. Hence every tested item eliminates the envy, as required.
\end{proof}
 
In particular, \Cref{thm:universal-ef1} already yields a complete \(\EFXpm\) allocation for every unnormalized Boolean instance, and each one-sided guarantee established below is automatically an \(\EFXpm\) guarantee.
 
\subsection{\texorpdfstring{The goods side: $\EFoneplus$}{The goods side: EF1+}}

We now state the fairness and efficiency guarantees of the allocations returned by \Cref{alg:goodsopt}.
\begin{restatable}{lemma}{goodsOPT}
\label{ilem:plus}
Every allocation $A$ returned by \Cref{alg:goodsopt} is Pareto optimal. If moreover $Z\neq\varnothing$, then $A$ is both $\EFoneplus$ and $\EFXzm$ (which further imply EF1 and $\EFXpm$).
\end{restatable}
\begin{proof}
Throughout, let $A$ be an allocation returned by \Cref{alg:goodsopt}. Pareto optimality follows from the fact that if $B$ Pareto dominates $A$, then $v_i(B_i)\ge v_i(A_i)$ for all $i$
with one strict inequality, so $W(B)>W(A)$, contradicting the definition of $\OPT^{+}(v)$.
 
Now we prove the fairness guarantees. Suppose $Z\neq\varnothing$ and let $(a,j)$ be an envious ordered pair, so that
\begin{equation}\label{eq:iviol}
  v_a(A_a)=0\qquad\text{and}\qquad v_a(A_j)=1.
\end{equation}
We establish two claims: (i) $v_a(A_j\setminus\{e\})=0$ for \emph{every} $e\in A_j$, and (ii) $A_j\neq\varnothing$. Each claim, when violated, produces an allocation that \Cref{alg:goodsopt} would strictly prefer to $A$, that is, an allocation with larger $W$, or with the same $W$ and smaller $C$; the only exception is one configuration, which forces $Z=\varnothing$.
 
\emph{Claim (i): deleting any item of the envied bundle eliminates the envy.} Suppose some
$e\in A_j$ has $v_a(A_j\setminus\{e\})=1$. Reallocate
$A'_a=A_j\setminus\{e\}$, $A'_j=A_a\cup\{e\}$, other bundles fixed. Then
$v_a(A'_a)=1$, so by \Cref{eq:iviol},
$W(A')-W(A)=1+v_j(A_a\cup\{e\})-v_j(A_j)$. If $v_j(A_j)=0$, or if $v_j(A_j)=v_j(A_a\cup\{e\})=1$,
the change in $W$ is positive, contradicting maximality of $W$ achieved by $A$. In the only other case, $v_j(A_j)=1$ and
$v_j(A_a\cup\{e\})=0$, welfare is unchanged while the joint contribution of $a,j$ to $C$ changes by
$(|A_j|-1)+0-(0+|A_j|)=-1$, so $C$ strictly drops, again a contradiction.
 
\emph{Claim (ii): the envied bundle is nonempty.} Suppose $A_j=\varnothing$. We first locate agent $a$
in $E$ with a nonempty bundle, then show that every other agent holds the empty bundle, and finally
conclude that $Z=\varnothing$, contrary to the hypothesis.

From \Cref{eq:iviol},
$v_a(\varnothing)=v_a(A_j)=1$, so $a\in E$; and $v_a(A_a)=0$ gives $A_a\neq\varnothing$.
 
Now let $k\neq a$. If $A_k\neq\varnothing$, dump $A_a$ onto $k$:
$A'_a=\varnothing$, $A'_k=A_k\cup A_a$. As $a\in E$, $v_a(A'_a)=1$, so
$W(A')-W(A)=1+v_k(A_k\cup A_a)-v_k(A_k)$. This is positive unless $v_k(A_k)=1$ and
$v_k(A_k\cup A_a)=0$, in which case welfare is unchanged and the joint contribution of $a,k$ to $C$
changes by $0-|A_k|<0$. Either way we contradict the optimality of $A$. Hence $A_k=\varnothing$, for all $k\neq a$.
 
Finally, fix any $k\neq a$, so that $A_k=\varnothing$ by the previous paragraph. If $v_k(\varnothing)=0$, swap $a$
and $k$: $A'_a=\varnothing$, $A'_k=A_a$. Then $v_a(A'_a)=1$ (up from $0$) and
$v_k(A_a)\ge0=v_k(\varnothing)$, so $W(A')-W(A)\ge1>0$, a contradiction. Thus every $k\neq a$ lies in
$E$, and with $a\in E$ this gives $Z=\varnothing$, contradicting the hypothesis. This proves Claim (ii).
 
By Claim (ii) the envied bundle $A_j$ is nonempty, and by Claim (i) any one of
its items is an $\EFoneplus$ witness; hence $A$ is $\EFoneplus$. For $\EFXzm$, the tested sets of the
envious pair $(a,j)$ are $P_{aj}=S_a^+(A_j)\cup S_a^0(A_j)$ and $O_a=S_a^-(A_a)$. By Claim (i)
every item of $A_j$ has deletion marginal $+1$ for $a$, so $P_{aj}=A_j$, which is nonempty by
Claim (ii), and every $e\in P_{aj}$ eliminates the envy by Claim (i). Every $e\in O_a$ eliminates
it automatically: $e\in S_a^-(A_a)$ means $v_a(A_a\setminus\{e\})=1\ge v_a(A_j)$. Hence $A$ is
$\EFXzm$, and by \Cref{prop:efx-ef1} also EF1 and $\EFXpm$.
\end{proof}

%The role of the assumption $Z\neq\varnothing$ is to
%exclude an empty envied bundle, which the proof handles using %one agent of $Z$.

Note that \Cref{ilem:plus} recovers the $\EFXzm$ guarantee of B\'erczi et al.~\cite[Theorem~7]{berczi2024} for the
normalized domain $\cVz^n$, and extends it to every profile with $Z\ne\varnothing$. Next we show that the assumption $Z\neq\varnothing$ in \Cref{ilem:plus} cannot be dropped: on the $Z = \varnothing$
domain, the $\EFoneplus$ guarantee fails not merely for the allocation returned by \Cref{alg:goodsopt} but for every allocation.
 
\begin{restatable}{proposition}{NoEFoneplus}\label{iprop:noplus}
If $Z=\varnothing$, there are instances with no $\EFoneplus$ allocation.
\end{restatable}
\begin{proof}
Take $n=2$, $m=1$, and $v_i(S)=\mathbf 1[S=\varnothing]$ for both agents (so both lie in $E$). The single
item goes to some agent $k$, giving $v_k(A_k)=0$, while the other agent holds $\varnothing$, which
is robust for $v_k$ and valued $1$ by $v_k$. By \Cref{ilem:viol}(a), this ordered pair violates
$\EFoneplus$, and this holds for every allocation.
\end{proof}

\subsection{\texorpdfstring{The chores side: $\EFoneminus$}{The chores side: EF1-}}

\Cref{ilem:minus-all-E} establishes the fairness guarantees of the allocations returned by \Cref{alg:choresopt}.
 
\begin{restatable}{lemma}{choreSide}\label{ilem:minus-all-E}
If $Z=\varnothing$, then every allocation $A$ returned by \Cref{alg:choresopt} is $\EFoneminus$, and hence also $\EFXpm$.
\end{restatable}
\begin{proof}
Suppose the pair $(a,j)$ violates $\EFoneminus$ in allocation $A$. By \Cref{ilem:viol}(b), $v_a'(A_a)=0$ and
$v_a(A_j)=1$. Since $a\in E$ and every bundle $S$ of size at most $1$ has $v_a'(S)=1$, the
condition $v_a'(A_a)=0$ forces $|A_a|\ge 2$. 
 
Fix $c\in A_a$. Reallocate agent $a$'s bundle to $A_j\cup\{c\}$ and give the remainder to $j$:
\[
  A'_a=A_j\cup\{c\},\qquad A'_j=A_a\setminus\{c\},\qquad\text{other bundles fixed.}
\]
Deleting $c$ from $A'_a$ leaves $A_j$, which $a$ values, so $v_a'(A'_a)=1$; thus,\[
  W'(A')-W'(A)=1+v_j'(A'_j)-v_j'(A_j).
\]
If $v_j'(A_j)=0$, or if $v_j'(A_j)=1=v_j'(A'_j)$, the increase in $W'$ is positive, contradicting maximality of
$W'$. In the remaining case $v_j'(A_j)=1$ and $v_j'(A'_j)=0$, welfare $W'$ is unchanged and we
compare $C'$, which $A$ maximizes. The joint contribution of $a,j$ to $C'$ changes by
\[
  \underbrace{|A_j|+1}_{v_a'(A'_a)=1,\ |A'_a|=|A_j|+1}+\underbrace{0}_{v_j'(A'_j)=0}
  \;-\;\Bigl(\underbrace{0}_{v_a'(A_a)=0}+\underbrace{|A_j|}_{ v_j'(A_j)=1}\Bigr)=1,
\]
so $C'$ strictly increases, again a contradiction. Hence no $\EFoneminus$ violation exists. Finally, $\EFoneminus$ implies EF1, which coincides with $\EFXpm$ by \Cref{prop:efx-ef1}.
\end{proof}

\begin{remark}
\label{rem:choresopt-no-efx}
Unlike its goods-side counterpart (\Cref{ilem:plus}), the guarantee of \Cref{ilem:minus-all-E} does
not upgrade to either of the two stronger EFX notions. Take $M=\{a,b,c\}$, let $v_2(S) = 1$ for all $S$, and let
$v_1(S)=1$ exactly for $S\in\{\varnothing,\{b\},\{c\}\}$, so $Z=\varnothing$. The allocation
$A=(\{a,b\},\{c\})$ lies in $\OPT^{-}(v)$ (\Cref{alg:choresopt}): both agents have $v_i'(A_i)=1$ (deleting $a$ from
$\{a,b\}$ leaves $\{b\}$), so $W'=n$ and $C'=m$ are both maximal. Agent $1$ envies agent $2$, and
the pair fails both notions. For $\EFXpz$, the item $b$ has deletion marginal $0$ in $A_1$, so it is
tested, yet $v_1(\{a\})=0<v_1(A_2)$. For $\EFXzm$, the item $c$ has deletion marginal $0$ in $A_2$,
so it is tested, yet $v_1(A_1)=0<v_1(\varnothing)$. For $\EFXzm$ this is no loss, since the notion
can fail to exist when $Z=\varnothing$ (\Cref{thm:efx0minus}). But $\EFXpz$ does
exist when $Z=\varnothing$, as we establish in \Cref{thm:negative-efx}.

\end{remark}
 
\begin{restatable}{lemma}{POEFZOne}\label{ilem:minus-Z1}
If $|Z|=1$, there is a complete allocation that is Pareto optimal and envy-free; in particular it is
both $\EFoneplus$ and $\EFoneminus$.
\end{restatable}
\begin{proof}
Let $Z=\{z\}$ and give all items to $z$: $A_z=M$ and $A_i=\varnothing$ for $i\neq z$. Every $i\neq z$
lies in $E$ and holds $\varnothing$ at value $1$, the maximum, so envies no one; and $z$ has
$v_z(\varnothing)=0\le v_z(M)$, so $z$ envies no one. Thus $A$ is envy-free. If $A$ is Pareto optimal,
we are done. Otherwise some $B$ Pareto dominates $A$; every $i\neq z$ already attains value $1$, so
$v_i(B_i)=1$ there, and the strict gain must come from $z$, giving $v_z(B_z)=1$. Then all agents
attain value $1$ under $B$, so $B$ is envy-free, and it is Pareto optimal because no value can
exceed $1$.
\end{proof}

\begin{remark}
\label{rem:Aplus-Z1}
When $|Z|=1$, every allocation in $\OPT^{+}(v)$, the optimal set of \Cref{alg:goodsopt}, is in fact envy-free. Let $A\in\OPT^{+}(v)$ and $Z=\{z\}$. Giving $M$ to $z$ and $\varnothing$ to every agent in $E$ yields welfare at least $n-1$. If the maximum welfare is $n$, then $A$ gives every agent value one and is envy-free. Otherwise, the maximum welfare is $n-1$, so $v_z(M)=0$ and the benchmark allocation attains the pair $(n-1,0)$; the secondary objective then gives $C(A)=0$, so every satisfied agent has an empty bundle. Because $v_z(\varnothing)=0$, agent $z$ is the unique unsatisfied agent; hence every agent in $E$ is satisfied with the empty bundle and completeness forces $A_z=M$. Agent $z$ values her own bundle and every other bundle at zero, while every agent in $E$ already attains value one. Thus the allocation is envy-free.
\end{remark}

The next proposition characterizes where $\EFoneminus$ fails to exist.
 
\begin{restatable}{proposition}{propNegativeEFOneMinus}\label{iprop:nominus}
If $|Z|\ge 2$, there are instances with no $\EFoneminus$ allocation.
\end{restatable}

In the following proposition, we identify the condition where Pareto optimality is incompatible with two-sided EF1.
 
\begin{restatable}{proposition}{POEFoneCounterexample}
\label{prop:po-impossibility}
For every \(n\ge2\) and \(m\ge2\), there is a profile in \(\cVo^n\) at which no allocation is both Pareto optimal and EF1.
\end{restatable}
\begin{proof}
Let every agent have the identical valuation \(v(S)=\ind{S=\varnothing}\). We first characterize the Pareto-optimal allocations. If two agents hold nonempty bundles, moving one entire bundle to the other agent makes its former owner empty and therefore strictly better off, while the recipient remains at value zero and every other agent is unaffected. Thus a Pareto-optimal allocation has exactly one nonempty bundle, which by completeness equals \(M\). Conversely, such an allocation is Pareto optimal: improving the unique holder requires making her bundle empty, which forces some currently empty agent to receive an item and fall from value one to zero.
 
Fix a Pareto-optimal allocation, and let \(k\) hold \(M\). Since \(m\ge2\), the bundle \(M\) and every one-item deletion from it are nonempty, so \(M\) is unsafe for \(v_k\). Every other agent holds \(\varnothing\), which is robust for \(v_k\). \Cref{lem:violation} therefore gives an EF1 violation from \(k\) to every other agent.
\end{proof}

 Finally, we combine the lemmas and propositions presented above to prove \Cref{ithm:trich}.

\subsection{Proof of \Cref{ithm:trich}} 
\begin{proof}
Each table entry is established by a preceding result.
 
\emph{$\EFoneplus$ and $\textup{PO}\wedge\EFoneplus$.} If $Z\neq\varnothing$, then
\Cref{ilem:plus} provides an allocation that is simultaneously Pareto optimal and
$\EFoneplus$; this establishes both rows in the columns $|Z|=1$ and $|Z|\ge2$. If
$Z=\varnothing$, then \Cref{iprop:noplus} exhibits an instance admitting no $\EFoneplus$ allocation, and therefore no Pareto-optimal $\EFoneplus$ allocation; this establishes both
rows in the column $|Z|=0$.
 
\emph{$\EFoneminus$.} If $|Z|=0$, then \Cref{ilem:minus-all-E} shows the existence of an
$\EFoneminus$ allocation; if $|Z|=1$, then \Cref{ilem:minus-Z1} provides an $\EFoneminus$
(indeed envy-free) allocation. If $|Z|\ge2$, then \Cref{iprop:nominus} exhibits an
instance admitting no $\EFoneminus$ allocation.
 
\emph{$\textup{PO}\wedge\EFoneminus$.} If $|Z|=1$, then the allocation of \Cref{ilem:minus-Z1} is Pareto
optimal and envy-free, hence $\EFoneminus$. If $|Z|=0$, then \Cref{prop:po-impossibility} shows that on $\cVo^n$ no
Pareto-optimal allocation is even $\EFone$; since $\EFoneminus$ implies $\EFone$, no
Pareto-optimal allocation is $\EFoneminus$ either. If $|Z|\ge2$, then $\EFoneminus$ itself may not exist by
\Cref{iprop:nominus}.
\end{proof}

\section{The EFX Landscape}
\label{sec:efx}
 
\begin{table}[htb!]
\centering
\small
\renewcommand{\arraystretch}{1.2}
% Changed \textwidth to \columnwidth, and adjusted column proportions
\begin{tabularx}{\columnwidth}{@{}l p{0.35\columnwidth} >{\raggedright\arraybackslash}X@{}}
\toprule
Notion & Status on arbitrary Boolean profiles & Source of the conclusion\\
\midrule
\(\EFXpm\) & Always exists & \Cref{thm:universal-ef1}, \Cref{prop:efx-ef1}\\
\(\EFXzm\) & Exists whenever \(Z\ne\varnothing\), but not universally & \Cref{ilem:plus}, \Cref{thm:efx0minus}\\
\(\EFXpz\) & Exists if \(\ Z = \emptyset \), but not universally & \Cref{thm:negative-efx}; counterexample by B\'erczi et al.~\cite{berczi2024}\\
\(\EFXzz\) & No universal guarantee & Counterexample by B\'erczi et al.~\cite{berczi2024}\\
\bottomrule
\end{tabularx}
\caption{The EFX landscape on arbitrary Boolean profiles.}
\label{tab:efx-status}
\end{table}
 
This section studies the existence of the four variants of EFX for Boolean valuations. In particular, we show that \(\EFXpz\) exists for negative-Boolean valuations, and show that \(\EFXzm\) exists if and only if \(Z\) is nonempty. The weakest notion is already resolved: by \Cref{prop:efx-ef1}, \(\EFXpm\) coincides with EF1, so \Cref{thm:universal-ef1} yields a complete \(\EFXpm\) allocation for every Boolean instance. Moreover, \Cref{ilem:plus} already provides a complete \(\EFXzm\) allocation whenever \(Z\ne\varnothing\), via \Cref{alg:goodsopt}.
 
\subsection{\texorpdfstring{\(\EFXpz\)}{EFX+0} for negative-Boolean valuations}
 
B\'erczi et al.\ prove \(\EFXpz\) existence for identical negative-Boolean valuations~\cite[Theorem~9]{berczi2024} and explicitly pose as an open question whether the assumption of identical valuations can be dropped~\cite[Question~16]{berczi2024}. The question remained open: Bhaskar et al.\ later obtained the weaker EF1 guarantee for negative-Boolean valuations~\cite[Theorem~3]{bhaskar2025}, but not \(\EFXpz\). \Cref{thm:negative-efx} below answers the question affirmatively.
 
The algorithm of B\'erczi et al.\ repairs a violation by transferring an item from the envious agent to the envied agent. Under nonidentical valuations, however, a repair for one agent can immediately create the reverse violation for the other. Our proof proceeds differently. We fix one reference agent and split the ground set into blocks that are minimal for that agent, in the sense that deleting any single item from a block raises its value; a block can then be given safely to any agent who values it at one, so we look for a matching that assigns every block to such an agent. If the matching exists, the resulting allocation is already \(\EFXpz\). If it does not exist, then Hall's condition fails, and the set of blocks witnessing the failure has a smaller neighborhood than itself; these blocks can be distributed among the reference agent and the agents in that neighbourhood, and every remaining agent values each of these blocks at zero. The remaining agents and items therefore form a strictly smaller instance whose solution can be combined with the partial allocation without creating any new violation, which is what drives the induction.

We keep the convention of the rest of the paper: the valuations are Boolean, and the domain of this subsection is \(\cVo^n\), that is, \(v_i(\varnothing)=1\) for every agent, or equivalently \(Z=\varnothing\). This is exactly the heterogeneous negative-Boolean domain of B\'erczi et al.\ after the translation \(u_i=v_i-1\), which sends our values \(1\) and \(0\) to their values \(0\) and \(-1\); as noted in \Cref{sec:preliminaries}, such a translation preserves every bundle comparison and every deletion marginal, and hence every notion of \Cref{sec:efx-notions}.

For \(v_i\in\cVo\), call a bundle \(S\) \emph{deletion-secure for agent \(i\)} if either \(v_i(S)=1\), or \(v_i(S)=0\) and \(v_i(S\setminus\{e\})=1\) for every \(e\in S\). Thus every value-one bundle is deletion-secure, as is every zero-valued bundle whose every one-item deletion has value one; in particular, every robust bundle is deletion-secure. The following lemma characterizes \(\EFXpz\) on this domain in terms of deletion-secure bundles.%; its proof can be found in~\Cref{app:efx}.

\begin{restatable}{lemma}{EFXPZCharacterization}
\label{lem:negative-efx-characterization}
An allocation \(A\) of an instance with \(v_i\in\cVo\) for every agent \(i\) is \(\EFXpz\) if and only if, for every agent \(i\), either \(A_i\) is deletion-secure for \(i\), or \(v_i(A_j)=0\) for every \(j\ne i\).
\end{restatable}
\begin{proof}
Fix an agent \(i\). If \(v_i(A_i)=1\), then \(i\) envies no agent, and \(A_i\) is deletion-secure for \(i\); both sides of the stated equivalence hold for \(i\). So assume \(v_i(A_i)=0\). If \(v_i(A_j)=0\) for every \(j\ne i\), then again \(i\) envies no agent and contributes no violation, matching the second alternative.

It remains to treat the case \(v_i(A_i)=0\) and \(v_i(A_j)=1\) for some \(j\ne i\), that is, \(i\) envies \(j\). For every \(e\in A_i\), the deletion marginal \(v_i(A_i)-v_i(A_i\setminus\{e\})\) is \(0\) or \(-1\), so \(A_i=S_i^-(A_i)\cup S_i^0(A_i)\) and every item of \(A_i\) is tested by \(\EFXpz\); such an item eliminates the envy precisely when \(v_i(A_i\setminus\{e\})=1\). Moreover, \(A_i\ne\varnothing\), because \(v_i(\varnothing)=1\ne0=v_i(A_i)\), so the tested set is nonempty. Finally, every \(e\in S_i^+(A_j)\) satisfies \(v_i(A_j\setminus\{e\})=0=v_i(A_i)\), so every tested item of \(A_j\) eliminates the envy automatically. Consequently the pair \((i,j)\) satisfies the requirement of \(\EFXpz\) if and only if \(v_i(A_i\setminus\{e\})=1\) for every \(e\in A_i\), which, since \(v_i(A_i)=0\), is exactly the statement that \(A_i\) is deletion-secure for \(i\). As this condition does not depend on \(j\), the agent \(i\) contributes no violation if and only if \(A_i\) is deletion-secure for \(i\), or \(i\) envies nobody, i.e.\ \(v_i(A_j)=0\) for every \(j\ne i\).
\end{proof}

\begin{theorem}\label{thm:negative-efx}
Every instance with \(v_i\in\cVo\) for every agent \(i\), that is, every instance with negative-Boolean valuations, admits a complete \(\EFXpz\) allocation. Such an allocation can be computed in polynomial time in the value-oracle model.\footnote{In the value-oracle model, we can query an agent $i$ with a bundle $S \subseteq M$ to get $v_i(S)$ in $O(1)$ time.}
\end{theorem}

\begin{proof}
We argue by induction on the number of agents \(n\). If \(n=1\), give all items to the single agent; there is no second agent, so the allocation is trivially \(\EFXpz\). Let \(n\ge2\) and fix an agent \(d\in N\), called the \emph{reference agent}.

We first cut the ground set into blocks that are minimal for the reference agent, in the sense that deleting any single item from a block raises its value for \(d\). Formally, we construct pairwise disjoint nonempty blocks \(C_1,\ldots,C_k\subseteq M\) and residual sets \(M=R_0\supsetneq R_1\supsetneq\cdots\supsetneq R_k\) for a suitable $k$, using an iterative procedure defined below.

We start with $R_0 = M$. For \(t=1,2,\ldots\): if \(R_{t-1}\) is deletion-secure for \(d\), stop and set \(k=t-1\). Otherwise \(R_{t-1}\) is not deletion-secure for \(d\), so \(v_d(R_{t-1})=0\) and there is an item \(e\in R_{t-1}\) with \(v_d(R_{t-1}\setminus\{e\})=0\). Starting from \(S=R_{t-1}\setminus\{e\}\), repeatedly replace \(S\) by \(S\setminus\{f\}\) for some \(f\in S\) with \(v_d(S\setminus\{f\})=0\), for as long as such an item \(f\) exists. The loop ends at a set \(C_t\subseteq R_{t-1}\setminus\{e\}\) with $v_d(C_t)=0$ and $
v_d(C_t\setminus\{f\})=1$ for every $f\in C_t,$ so that \(C_t\) is deletion-secure for \(d\); note that \(C_t\ne\varnothing\), since \(v_d(\varnothing)=1\ne0\). Set \(R_t=R_{t-1}\setminus C_t\) and continue.

Each block is nonempty, so each iteration strictly shrinks the residue and the loop terminates after \(k\le m\) iterations. On termination, \(M\) is the disjoint union of \(C_1,\ldots,C_k\) and \(R_k\), we have \(v_d(C_t)=0\) for every \(t\le k\), and \(R_k\) is deletion-secure for \(d\).

We next record which agents can safely receive which blocks. Put \(p=\min\{k,n-1\}\), and let \(H\) be the bipartite graph with parts \(\mathcal L=\{C_1,\ldots,C_p\}\) and \(N\setminus\{d\}\), where a block \(C_t\) and an agent \(i\) are adjacent exactly when \(v_i(C_t)=1\). Whether \(H\) admits a matching that saturates \(\mathcal L\) decides how we proceed, and we consider these two cases separately.

{\bf Case $(i)$.} Suppose first that some matching of \(H\) saturates \(\mathcal L\). Fix such a matching, give each block \(C_t\in\mathcal L\) to the agent matched to it, give \(R_p\) to \(d\), and give the empty bundle to every remaining agent. The bundles \(C_1,\ldots,C_p,R_p\) are pairwise disjoint with union \(M\), so the allocation is complete. Every agent matched to a block \(C_t\) has \(v_i(C_t)=1\), and every agent receiving \(\varnothing\) has \(v_i(\varnothing)=1\); in both cases the bundle is deletion-secure for its owner. For the reference agent there are two possibilities. If \(p=k\), then \(d\) receives \(R_k\), which is deletion-secure for \(d\) by the construction above. If \(p=n-1<k\), then every agent other than \(d\) is matched to a block, and \(v_d(C_t)=0\) for every block, so \(v_d(A_j)=0\) for every \(j\ne d\) and the second alternative of \Cref{lem:negative-efx-characterization} applies to \(d\) even when \(R_p\) is not deletion-secure. In either case \Cref{lem:negative-efx-characterization} shows that the allocation is \(\EFXpz\).

{\bf Case $(ii)$.} It remains to treat the case in which no matching of \(H\) saturates \(\mathcal L\) (induction hypothesis will get utilized here). Here the blocks that witness the failure themselves split the instance into a part we can allocate at once and a strictly smaller part that we solve recursively. Fix a maximum matching \(F\) of \(H\) and a block \(C^\star\in\mathcal L\) that \(F\) leaves unmatched. Call a path of \(H\) \emph{\(F\)-alternating from \(C^\star\)} if it starts at \(C^\star\), its first edge lies outside \(F\), and its edges thereafter lie alternately in \(F\) and outside \(F\); the one-vertex path \(C^\star\) is \(F\)-alternating from \(C^\star\) as well. Let
\begin{align*}
X&=\{C\in\mathcal L: \text{some \(F\)-alternating path from \(C^\star\) ends at } C\},\\
Q&=\{i\in N\setminus\{d\}: \text{some \(F\)-alternating path from \(C^\star\) ends at } i\},
\end{align*}
so that \(C^\star\in X\). Every agent of \(Q\) is matched by \(F\): an \(F\)-alternating path from \(C^\star\) ending at an unmatched agent would be an augmenting path, contradicting the maximality of \(F\).

We claim that \(F\) matches the agents of \(Q\) with the blocks of \(X\setminus\{C^\star\}\), one to one. Let \(i\in Q\), and let \(C\) be the block matched to \(i\) by \(F\), which exists by the previous paragraph. Appending the edge \(\{i,C\}\in F\) to an \(F\)-alternating path from \(C^\star\) ending at \(i\) gives an \(F\)-alternating path from \(C^\star\) ending at \(C\), so \(C\in X\); and \(C\ne C^\star\) because \(C^\star\) is unmatched. This assigns to every \(i\in Q\) a block of \(X\setminus\{C^\star\}\), and distinct agents receive distinct blocks because \(F\) is a matching. Conversely, let \(C\in X\setminus\{C^\star\}\) and consider an \(F\)-alternating path from \(C^\star\) ending at \(C\); since the path has at least one edge and alternates starting outside \(F\), its last edge lies in \(F\) and joins \(C\) to an agent \(i\), which lies in \(Q\) because the path truncated before \(C\) ends at \(i\). Hence every block of \(X\setminus\{C^\star\}\) is assigned to some agent of \(Q\), and the assignment is a bijection; in particular \(|Q|=|X|-1\).

We claim next that the set of agents adjacent in \(H\) to at least one block of \(X\) is exactly \(Q\). Every agent of \(Q\) is adjacent to the block matched to it, which lies in \(X\), so one inclusion is clear. For the other, let \(C\in X\) and let \(i\) be an agent adjacent to \(C\). If \(\{i,C\}\notin F\), take an \(F\)-alternating path from \(C^\star\) ending at \(C\); its last edge lies in \(F\), or the path is the one-vertex path \(C^\star\), so appending \(\{C,i\}\) preserves the alternation and shows \(i\in Q\). If \(\{i,C\}\in F\), then \(C\ne C^\star\), and by the bijection of the previous paragraph the block \(C\) is matched to an agent of \(Q\), which must be \(i\) because \(F\) is a matching. In both cases \(i\in Q\).

Now allocate as follows. Give \(C^\star\) to \(d\), and give every block \(C\in X\setminus\{C^\star\}\) to the agent of \(Q\) matched to it by \(F\); by the bijection above, each agent of \(Q\) receives exactly one block. Write
\[
P=Q\cup\{d\},
\qquad
M_P=\bigcup_{C\in X}C
\]
for the set of agents served so far and the set of items they receive. Every agent of \(P\) holds a bundle that is deletion-secure for her: the bundle \(C^\star\) is a peeled block and hence deletion-secure for \(d\), and each \(i\in Q\) receives a block \(C\) with \(v_i(C)=1\).

Consider the residual instance with agent set \(N'=N\setminus P\) and item set \(M'=M\setminus M_P\), in which every agent \(i\in N'\) keeps the restriction of \(v_i\) to the subsets of \(M'\); this restriction still assigns the value \(1\) to \(\varnothing\), so the residual instance again has all valuations in \(\cVo\). Since \(|Q|+1=|X|\le p\le n-1\), we have \(|N'|=n-|X|\ge1\), and \(|N'|<n\) because \(d\in P\). By the induction hypothesis, the residual instance admits a complete \(\EFXpz\) allocation \((A_i)_{i\in N'}\) of \(M'\). Let \(A\) be the allocation of \(M\) among \(N\) that gives every agent of \(P\) the bundle assigned above and every agent of \(N'\) her bundle in the residual allocation. Since \(M_P\) and \(M'\) partition \(M\), \(A\) is complete.

We verify the condition of \Cref{lem:negative-efx-characterization} for every agent of \(N\), which proves that \(A\) is \(\EFXpz\). For \(i\in P\), the bundle \(A_i\) is deletion-secure for \(i\), as observed above. Let now \(i\in N'\). Every bundle held by an agent of \(P\) is a block of \(X\), and \(i\notin Q\), so \(i\) is adjacent in \(H\) to no block of \(X\) by the previous claim; that is,
\begin{equation}\label{eq:no-cross-envy}
v_i(A_j)=0\qquad\text{for every } j\in P .
\end{equation}
If \(A_i\) is deletion-secure for \(i\), the first alternative of the characterization holds for \(i\) and we are done. Otherwise \(A_i\) is not deletion-secure for \(i\); applying \Cref{lem:negative-efx-characterization} to the residual allocation, which is \(\EFXpz\) by the induction hypothesis, yields \(v_i(A_j)=0\) for every \(j\in N'\setminus\{i\}\), and together with \Cref{eq:no-cross-envy} this gives \(v_i(A_j)=0\) for every \(j\in N\setminus\{i\}\), which is the second alternative. Note that the correctness of the induction depends upon \Cref{eq:no-cross-envy}: the bundles allocated to \(P\) are worth zero to every agent of \(N'\), so no agent of the residual instance envies an agent of \(P\), and the recursive allocation can be computed without any reference to the bundles already assigned.

{\bf Running time.} Deciding whether a set is deletion-secure for \(d\) takes at most \(m+1\) value queries, and each block is produced by at most \(m\) deletions, so producing all the blocks uses \(O(m^3)\) value queries. The graph \(H\) is built with at most \(nm\) queries. A maximum matching of \(H\), and the vertex sets \(X\) and \(Q\), are computable in polynomial time by standard algorithms for maximum matching. Every recursive call removes at least one agent, namely \(d\), so there are at most \(n\) calls. Hence the algorithm runs in polynomial time.
\end{proof}

\subsection{Tightly characterizing the existence of \texorpdfstring{\(\EFXzm\)}{EFX0-}}
 
\begin{restatable}{theorem}{efxZminus}
\label{thm:efx0minus}
Every instance with Boolean valuations where \(Z=\{i:v_i(\varnothing)=0\}\ne\varnothing\) admits a complete \(\EFXzm\) allocation. This is best possible as a guarantee at the level of valuation classes: there is an instance with \(Z=\varnothing\) that admits no \(\EFXzm\) allocation.
\end{restatable}
\begin{proof}
Existence of \(\EFXzm\) when \(Z\ne\varnothing\) follows from \Cref{ilem:plus}. For the negative result, take two agents, three items, and the identical valuation \(v(S)=1\) if and only if \(|S|\le1\). Every allocation has a bundle-size split of \((0,3)\) or \((1,2)\), up to exchanging the agents.
 
If the bundle sizes are \((0,3)\), the three-item holder has value zero and envies the empty bundle, which has value one. The envied bundle has no item, and deletion of any item from the three-item bundle leaves value zero. Thus the tested set \(S_i^+(A_j)\cup S_i^0(A_j)\cup S_i^-(A_i)\) is empty, violating $\EFXzm$.
 
If the bundle sizes are \((1,2)\), the two-item holder has value zero and envies the singleton, which has value one. The singleton item has zero deletion marginal because deleting it leaves the empty bundle, also of value one. It is therefore tested under \(\EFXzm\), but its deletion does not eliminate the envy. Hence neither bundle-size pattern is \(\EFXzm\).
\end{proof}

\subsection{The remaining variants of EFX}
 
No universal guarantee on arbitrary Boolean profiles is possible for either \(\EFXpz\) or \(\EFXzz\). B\'erczi et al.\ give a two-agent, three-item instance with identical positive-Boolean valuation \(v(S)=\ind{|S|\ge2}\) with no \(\EFXpz\) allocation, and a two-agent, two-item identical additive Boolean instance with singleton values one and zero with no \(\EFXzz\) allocation~\cite{berczi2024}. The first counterexample complements \Cref{thm:negative-efx}: \(\EFXpz\) can fail when \(|Z|=n\), but always exists when \(|Z|=0\). \Cref{tab:efx-status} summarizes the resulting landscape.

\section{Combining Fairness and Efficiency with Incentives}
\label{sec:mechanisms}
In the setting where every agent values the empty bundle at zero ($|Z| = n$), B\'erczi et al.~\cite{berczi2024} showed that \(\EFXzm\) allocations always exist. We strengthen this by designing a deterministic weakly group-strategyproof mechanism that always outputs PO and \(\EFXzm\) allocations.
  
The mechanism of this section is simply \Cref{alg:goodsopt} with a tie-breaking rule. Recall from \Cref{alg:goodsopt} the scores \(W_p(A)=\sum_{i\in N}p_i(A_i)\) and \(C_p(A)=\sum_{i\in N}p_i(A_i)|A_i|\) of an allocation \(A\in\cA\) for any profile of Boolean valuation functions $p=(p_1,\ldots,p_n) \in (\cVz \cup \cVo)^n$. These scores are, respectively, the number of agents reporting value one for their bundles and the total size of those agents' bundles. Fix a strict total order \(\rhd\) on the set of all allocations, independent of all reports; this tie-breaking order, which plays no role in the existence results of \Cref{sec:onesided}, is what turns \Cref{alg:goodsopt} into a deterministic mechanism. The mechanism \(\mathcal M_\rhd\), stated as \Cref{alg:mechanism}, runs \textsc{GoodsOpt} on the reported profile \(p\) and resolves its final choice by \(\rhd\): it maximizes \(W_p(A)\), minimizes \(C_p(A)\) among the maximizers, and returns the allocation that is maximal as per \(\rhd\) among the allocations that remain. In particular, \(\mathcal M_\rhd(p)\in\OPT^{+}(p)\).

\begin{algorithm}%[H]
\caption{The mechanism \(\mathcal M_\rhd\)}
\label{alg:mechanism}
\KwIn{A reported Boolean profile $p=(p_1,\ldots,p_n)$; a strict total order $\rhd$ on $\cA$, fixed independently of $p$.}
\KwOut{A complete allocation.}
$\cW\leftarrow\arg\max_{A\in\cA} W_p(A)$\;
$\OPT^{+}(p)\leftarrow\arg\min_{A\in\cW} C_p(A)$\;
\Return{the $\rhd$-maximal allocation in $\OPT^{+}(p)$}
\end{algorithm}
 
We show that this mechanism resists not only unilateral deviations by agents but also coordinated misreporting attempts by a group of agents.
 
\begin{definition}[Weak group strategyproofness]
\label{def:gsp}
A mechanism \(f\) is \emph{weakly group-strategyproof} if no coalition can misreport so that all of its members strictly benefit: for every nonempty coalition \(K\subseteq N\), every true profile \(v\in\cV^n\), and every joint misreport \(\widehat v_K=(\widehat v_i)_{i\in K}\), some member \(i\in K\) has
\[
v_i\bigl(f_i(\widehat v_K,\, v_{-K})\bigr)\;\le\; v_i\bigl(f_i(v)\bigr).
\]
Taking \(|K|=1\) recovers strategyproofness.
\end{definition}

%See~\Cref{app:mechanisms} for the proof of the following theorem.

The following theorem formalizes the guarantees achieved by~\Cref{alg:mechanism}.

\begin{restatable}{theorem}{mechanism}
\label{thm:score-mechanism}
If the agents are allowed to report any valuations from \(\cV^n\), the mechanism \(\mathcal M_\rhd\) is weakly group-strategyproof and returns a Pareto-optimal allocation. If the reports lie in \(\cVz^n\), its output additionally satisfies the fairness guarantees \(\EFoneplus\) and \(\EFXzm\), and hence \(\EFXpm\), with respect to the reported profile.
\end{restatable}
\begin{proof}
Fix a reported profile \(p\), and write \(A=\mathcal M_\rhd(p)\).
 
\emph{Pareto optimality and fairness.} Since \(A\in\OPT^{+}(p)\) is returned by \Cref{alg:goodsopt}, both claims are exactly \Cref{ilem:plus} applied to the profile \(p\): Pareto optimality holds unconditionally, and \(\EFoneplus\) together with \(\EFXzm\) (hence \(\EFXpm\), by \Cref{prop:efx-ef1}) holds since \(Z(p)=\{i:p_i(\varnothing)=0\}=N\).
 
\emph{Weak group strategyproofness.} Fix a nonempty coalition \(K\subseteq N\), a true profile \(v\in\cV^n\), and misreports \(\widehat v_i\in\cV\) for \(i\in K\); let \(\widehat v\) denote the profile that agrees with the misreports on \(K\) and with \(v\) elsewhere. Let \(X=\mathcal M_\rhd(v)\) and \(Y=\mathcal M_\rhd(\widehat v)\). A profitable joint deviation would require every member to gain: \(v_i(X_i)=0\) and \(v_i(Y_i)=1\) for all \(i\in K\); indeed this implies that \(Y\ne X\). Put \(w^*=W_v(X)\) and \(c^*=C_v(X)\).
 
For any allocation \(B\), changing only the coalition's reports changes its scores by
\[
W_{\widehat v}(B)-W_v(B)=\sum_{i\in K}\bigl(\widehat v_i(B_i)-v_i(B_i)\bigr)\]

\[C_{\widehat v}(B)-C_v(B)=\sum_{i\in K}\bigl(\widehat v_i(B_i)-v_i(B_i)\bigr)|B_i|;
\]
in particular, if \(\widehat v_i(B_i)=v_i(B_i)\) for every \(i\in K\), then \(B\) has the same pair of scores under both profiles. Since \(v_i(Y_i)=1\) for every member, each term \(\widehat v_i(Y_i)-v_i(Y_i)\) is nonpositive, whence
\[
W_{\widehat v}(Y)=W_v(Y)+\sum_{i\in K}\bigl(\widehat v_i(Y_i)-1\bigr)\le W_v(Y)\le w^*,
\]
with equality throughout only if \(\widehat v_i(Y_i)=1\) for every \(i\in K\) and \(W_v(Y)=w^*\).
 
If \(\widehat v_i(X_i)=1\) for some member \(i\), then, since \(v_j(X_j)=0\) for every \(j\in K\), the welfare of \(X\) can only increase: \(W_{\widehat v}(X)\ge w^*+1\). By the preceding bound, \(Y\) cannot be selected. Hence \(\widehat v_i(X_i)=0=v_i(X_i)\) for every member, and \(X\) has the same pair of scores \((w^*,c^*)\) under both profiles.
 
For \(Y\) to defeat \(X\) under \(\widehat v\), it cannot have strictly larger welfare by the preceding bound. If it has welfare \(w^*\), the bound is tight, which forces \(\widehat v_i(Y_i)=v_i(Y_i)=1\) for every \(i\in K\) and \(W_v(Y)=w^*\). The scores of \(Y\) are therefore also unchanged by the deviation. It cannot have \(C_{\widehat v}(Y)<c^*\), because then \(C_v(Y)<c^*\) would contradict the optimality of $X$. If the two score pairs are equal, the fixed order \(\rhd\) selects \(X\) at the truthful profile and therefore must rank \(X\) above \(Y\) at the deviating profile as well. Thus \(Y\) could not have been chosen by \(\mathcal M_\rhd\) on the profile \(\widehat v\), a contradiction.
\end{proof}

Weak group strategyproofness rules out joint deviations that strictly benefit every coalition member. Strong group strategyproofness instead rules out a joint deviation under which every coalition member is weakly better off and at least one is strictly better off. Prior work shows that strong group strategyproofness is incompatible with Pareto optimality even for unit-demand dichotomous valuations~\cite{babaioff2021dichotomous,barman-verma-truthful}, which form a subclass of \(\cVz\). Hence no deterministic mechanism on the normalized Boolean domain can satisfy both strong group strategyproofness and Pareto optimality, even without a fairness requirement. \Cref{example:notSGSP} in~\Cref{app:mechanisms} shows that $\mathcal{M}_\rhd(v)$ is not strongly group-strategyproof.

\section{Best of Both Worlds: Adding Ex-Ante Envy-Freeness}
\label{sec:bobw}
In this section we show that randomization strengthens fairness without sacrificing any of the guarantees obtained so far: whenever $|Z| \geq 1$, i.e., some agent values the empty bundle at zero, drawing an allocation uniformly at random from the optimal set of \Cref{alg:goodsopt} yields a lottery that is envy-free in expectation, while every allocation it may output is Pareto optimal, $\EFoneplus$, and $\EFXzm$.

The condition $|Z| \geq 1$ is tight as a class-wide guarantee. Indeed, for every $n,m\ge2$, \Cref{prop:po-impossibility} gives a profile with $|Z| = 0$ at which no allocation is simultaneously Pareto optimal and EF1; hence no lottery can have support consisting of Pareto-optimal $\EFoneplus$ allocations, regardless of its ex-ante properties. In that counterexample, the uniform lottery over the Pareto-optimal allocations is ex-ante envy-free by symmetry, so it is precisely the ex-post fairness requirement that fails.

Throughout this section we assume $|Z| \geq 1$. For an allocation $A$ and an ordered pair $i\neq j$ define the \emph{envy differential} $d_A(i,j)=v_i(A_i)-v_i(A_j)\in \{-1,0,1\}.$ Agent $i$ \emph{envies} $j$ in $A$ if $d_A(i,j)=-1$.

\begin{definition}[ex-ante EF; BoBW lottery]
A lottery (probability distribution) $p$ over allocations is \emph{ex-ante
envy-free} if $\mathbb{E}_{A\sim p}\bigl[v_i(A_i)\bigr]\;\ge\;\mathbb{E}_{A\sim p}\bigl[v_i(A_j)\bigr]$ for every ordered pair $i\neq j,$
equivalently $\mathbb{E}_{A\sim p}[d_A(i,j)]\allowbreak\ge 0$ for all $i\neq j$. A \emph{BoBW
lottery} is a lottery that is ex-ante EF and whose support consists of allocations
that are PO and $\EFoneplus$.
\end{definition}
  
The lottery is built directly on \textsc{GoodsOpt} (\Cref{alg:goodsopt}): instead of selecting a single allocation from the optimal set $\OPT^{+}(v)$, we randomize uniformly over the whole set. Recall the scores of \Cref{sec:optimizers} on the true profile: $W(A)=\sum_{i\in N} v_i(A_i)$ is the utilitarian welfare, and $C(A)=\sum_{i\in N}v_i(A_i)|A_i|$ is the number of items the agents having utility $1$ hold. Let $(W^{\ast},C^{\ast})$ be the optimal score pair, so that $\OPT^{+}(v)=\bigl\{A\in\cA:\ W(A)=W^{\ast}\ \text{and}\ C(A)=C^{\ast}\bigr\}$ is the family of allocations computed by \textsc{GoodsOpt} (\Cref{alg:goodsopt}); $\OPT^{+}(v)\neq\varnothing$.
 
\begin{definition}[the uniform optimizer lottery]\label{def:lottery}
The lottery $p^{\mathrm{unif}}$ is the uniform distribution over $\OPT^{+}(v)$: it draws
each $A\in\OPT^{+}(v)$ with probability $1/|\OPT^{+}(v)|$.
\end{definition}

The main result of this section is stated next.%; the proof is deferred to~\Cref{app:bobw}.
\begin{theorem}\label{thm:bobwZ}
Let $v_1,\dots,v_n\colon 2^{M}\to\{0,1\}$ be arbitrary Boolean valuations with
$Z = \{i : v_i(\varnothing) = 0\} \neq\varnothing$. Then $p^{\mathrm{unif}}$ (\Cref{def:lottery}) is a BoBW lottery: it is ex-ante
envy-free, and every allocation in its support is Pareto optimal, $\EFoneplus$, and $\EFXzm$.
\end{theorem}

Since the support of the lottery satisfies $\supp(p^{\mathrm{unif}})=\OPT^{+}(v)$, the ex-post guarantees of PO, $\EFoneplus$, and $\EFXzm$ directly follow from~\Cref{ilem:plus}. The remainder of this
section proves the ex-ante part.
  
For $i\neq j$, the \emph{swap operator} $T_{ij}$ maps an allocation $A$ to the
allocation $A'=T_{ij}(A)$ defined by
\[
A'_i=A_j,\qquad A'_j=A_i,\qquad A'_k=A_k\quad (k\notin\{i,j\}).
\]
Note that $T_{ij}$ is an involution ($T_{ij}$ is its own inverse) on the set of all allocations, i.e.,
$T_{ij}(T_{ij}(A))=A$.
 
\begin{lemma}[swap closure]\label{lem:swap}
Let $A\in\OPT^{+}(v)$ and suppose $d_A(i,j)=-1$. Then $A'=T_{ij}(A)\in\OPT^{+}(v)$ and
$d_{A'}(i,j)=+1$.
\end{lemma}
 
\begin{proof}
Since the valuations are Boolean, $d_A(i,j)=-1$ means $v_i(A_i)=0, v_i(A_j)=1$.
 
\medskip
First, we will show that the swap preserves welfare and pins down $j$'s values.
Only agents $i$ and $j$ change bundles, so
\begin{align*}
W(A')-W(A)
&=\bigl(v_i(A'_i)+v_j(A'_j)\bigr)-\bigl(v_i(A_i)+v_j(A_j)\bigr)\\
&=\bigl(v_i(A_j)+v_j(A_i)\bigr)-\bigl(v_i(A_i)+v_j(A_j)\bigr)\\
&=1+v_j(A_i)-v_j(A_j).
\end{align*}
Since $A$ attains the maximum welfare $W^{\ast}$, we must have
$W(A')-W(A)\le 0$, i.e.\ $v_j(A_j)-v_j(A_i)\ge 1$. As both values lie in $\{0,1\}$,
this forces $v_j(A_j)=1, v_j(A_i)=0$,
and consequently $W(A')=W(A)=W^{\ast}$, so $A'$ attains the maximum welfare.
 
\medskip
Now, we will prove that the swap preserves the secondary cost.
By the above argument,
\[
v_i(A'_i)=v_i(A_j)=1,\qquad v_j(A'_j)=v_j(A_i)=0 ,
\]
so among $\{i,j\}$ exactly $j$ is satisfied in $A$ and exactly $i$
is satisfied in $A'$. Every agent
$k\notin\{i,j\}$ keeps her bundle, hence her satisfaction status and her
contribution to $C$ are unchanged. Therefore
\[
C(A')-C(A)=|A'_i|-|A_j|=|A_j|-|A_j|=0,
\]
so $C(A')=C(A)=C^{\ast}$.

Hence, $A'$ attains the optimal score pair
$(W^{\ast},C^{\ast})$, and therefore $A'\in\OPT^{+}(v)$. Finally,
\[
d_{A'}(i,j)=v_i(A'_i)-v_i(A'_j)=v_i(A_j)-v_i(A_i)=1-0=+1. \qedhere
\]
\end{proof}
 
\begin{proposition}\label{prop:exante}
The lottery $p^{\mathrm{unif}}$ is ex-ante envy-free.
\end{proposition}
 
\begin{proof}
Fix an ordered pair $i\neq j$ and partition $\OPT^{+}(v)$ according to the value of
$d_A(i,j)$:
\[
S_{d}=\{A\in\OPT^{+}(v): d_A(i,j)=d\}
\qquad\text{for } d\in\{-1,0,+1\}.
\]
By \Cref{lem:swap}, the swap operator $T_{ij}$ maps $S_{-1}$ into $S_{+1}$.
Since $T_{ij}$ is an involution on the set of all allocations, it is injective, and
therefore
\[
|S_{-1}|\;\le\;|S_{+1}| .
\]
Consequently,
\begin{align}
&\mathbb{E}_{A\sim p^{\mathrm{unif}}}\bigl[v_i(A_i)-v_i(A_j)\bigr]\\
&=\frac{1}{|\OPT^{+}(v)|}\sum_{A\in\OPT^{+}(v)} d_A(i,j)
=\frac{|S_{+1}|-|S_{-1}|}{|\OPT^{+}(v)|}\;\ge\;0 .
\end{align}
As the ordered pair $(i,j)$ was arbitrary,
$\mathbb{E}_{A\sim p^{\mathrm{unif}}}[v_i(A_i)]\allowbreak \ge
\allowbreak  \mathbb{E}_{A\sim p^{\mathrm{unif}}}[\allowbreak v_i(A_j)]$ for all $i\neq j$; that is,
$p^{\mathrm{unif}}$ is ex-ante envy-free.
\end{proof}
 
\begin{proof}[Proof of \Cref{thm:bobwZ}]
The support of $p^{\mathrm{unif}}$ is $\OPT^{+}(v)$, and by \Cref{ilem:plus} every $A\in\OPT^{+}(v)$ is Pareto optimal, $\EFoneplus$, and $\EFXzm$; this establishes the ex-post guarantee. \Cref{prop:exante} establishes ex-ante envy-freeness. Hence $p^{\mathrm{unif}}$ is a BoBW lottery.
\end{proof}

\section{Matroid Constraints}
\label{sec:matroid}

Until now every partition of the items has been admissible. We now require the bundles to be independent in a matroid, which models settings such as course allocation with per-category quotas. For two agents we determine exactly which matroids admit a feasible EF1 allocation for every Boolean profile; the answer concerns the way in which one pair of complementary independent sets can be turned into another by exchanging single items, and through this route the fairness question meets two conjectures of matroid theory that have been open for decades.

\subsection{Matroid preliminaries}
 
Let \(\cK=(M,\cI)\) be the matroid. That is, \(M\) is the ground set of items and \(\cI\subseteq 2^M\) is the family of \emph{independent sets}, satisfying three axioms: $(i)$ the non-emptiness axiom which requires \(\varnothing\in\cI\), $(ii)$ the hereditary axiom which states that if \(X\in\cI\) and \(Y\subseteq X\), then \(Y\in\cI\), and $(iii)$ the augmentation axiom, with the requirement that if \(X,Y\in\cI\) with \(|X|<|Y|\), then some \(e\in Y\setminus X\) satisfies \(X\cup\{e\}\in\cI\).

A \emph{basis} is an inclusion-maximal independent set. The augmentation axiom implies that all bases have the same cardinality, called the \emph{rank} of the matroid. In a uniform matroid of rank \(r\), a set is independent exactly when it has cardinality at most \(r\). A partition matroid partitions \(M\) into blocks \(P_q\) with capacities \(r_q\), and a set \(S\) is independent exactly when \(|S\cap P_q|\le r_q\) for every block \(P_q\).

We assume that there is a common matroid \(\cK=(M,\cI)\) for every agent. An allocation \(A\) is \emph{feasible} if \(A_i\in\cI\) for every \(i\); the allocation continues to satisfy $\cup_i A_i = M$ and $A_i \cap A_j = \emptyset$ for all distinct $i,j \in N$. We assume that at least one complete feasible allocation exists.

Fairness is defined exactly as in \Cref{sec:preliminaries}: a complete feasible allocation \(A\) is EF1 if, for every ordered pair of agents \(i,j\), either \(v_i(A_i)\ge v_i(A_j)\), or some item \(e\in A_i\cup A_j\) satisfies \(v_i(A_i\setminus\{e\})\ge v_i(A_j\setminus\{e\})\). The constraint thus restricts which allocations are admissible, not which comparisons the definition may make; and since the independent sets of a matroid are closed under taking subsets, the bundles obtained by deleting a single item in the fairness test are themselves independent.

In the following subsection, we uncover the necessary and sufficient condition that the feasibility matroid $\cK$ must satisfy to enable the existence of feasible EF1 allocations.
 
\subsection{A characterization for two agents}
 
Let $\cF_{\cK}$ be the set of independent sets of $\cK$ whose complement sets are also independent sets, formally,
\[
\cF_{\cK}=\{B\subseteq M:B\in\cI\text{ and }M\setminus B\in\cI\}.
\]
Thus every \(B\in\cF_{\cK}\) represents the ordered feasible allocation \((B,M\setminus B)\). Define a graph \(H_{\cK}\) on vertex set \(\cF_{\cK}\), joining distinct vertices \(B,B'\) whenever
\[
\lowerstar B\cap\lowerstar B'\ne\varnothing
\qquad\text{or}\qquad
\lowerstar(M\setminus B)\cap\lowerstar(M\setminus B')\ne\varnothing.
\]
Define the \emph{complementation map}
\[
\tau:\cF_{\cK}\to\cF_{\cK},\qquad \tau(B)=M\setminus B .
\]
It is well defined, because the two conditions defining \(\cF_{\cK}\) are exchanged when \(B\) is replaced by \(M\setminus B\), and it satisfies \(\tau(\tau(B))=B\). Furthermore, \(\tau\) is an automorphism of \(H_{\cK}\): replacing \(B\) by \(M\setminus B\) and \(B'\) by \(M\setminus B'\) simultaneously exchanges the two conditions of the adjacency rule above, so \(B\) and \(B'\) are adjacent if and only if \(\tau(B)\) and \(\tau(B')\) are. Consequently \(\tau\) maps every connected component of \(H_{\cK}\) onto a connected component. For a component \(K \subseteq \cF_{\cK}\) of \(H_{\cK}\), we write \(\tau(K)=\{\tau(B):B\in K\}\), and we call \(K\) \emph{self-complementary} if \(\tau(K)=K\).
 
\begin{theorem}
\label{thm:matroid-characterization}
An instance with two agents and a common matroid feasibility constraint represented by $\cK$ admits a complete feasible EF1 allocation if and only if \(H_{\cK}\) has a connected component \(K\) satisfying \(\tau(K)=K\).
\end{theorem}
 
\begin{proof}
For sufficiency, let \(K\) be a self-complementary component and fix \(B\in K\). Since \(\tau(K)=K\), the vertex \(\tau(B)\) also lies in \(K\), so the connectedness of \(K\) provides a path \(B=X_0,X_1,\ldots,X_\ell=\tau(B)\) in \(H_{\cK}\). Call a vertex \(X\in\cF_{\cK}\) \emph{agreeable for \(v_1\)} if agent \(1\) is the source of no EF1 violation when she receives \(X\) and the other agent receives \(M\setminus X\), and also when she receives \(M\setminus X\) and the other agent receives \(X\). We claim that the path contains an agreeable vertex.

Recall from \Cref{sec:preliminaries} that a bundle \(S\) is \emph{unsafe} for \(v_1\) if \(v_1(T)=0\) for every \(T\in\lowerstar S\), and \emph{robust} for \(v_1\) if \(v_1(T)=1\) for every \(T\in\lowerstar S\). By \Cref{lem:violation}, agent \(1\) holding \(S\) while the other agent holds \(M\setminus S\) is an EF1 violation exactly when \(S\) is unsafe and \(M\setminus S\) is robust for \(v_1\).

Assume, for contradiction, that no vertex of the path is agreeable. Consider a vertex \(X_q\) along the path. Exactly one of the two ways of assigning the pair \(\{X_q,M\setminus X_q\}\) to the two agents makes agent \(1\) the source of a violation: at least one by the assumption, and not both, because a bundle cannot be at once unsafe and robust for \(v_1\). We may therefore label \(X_q\) by \(0\) if \(X_q\) is unsafe and \(M\setminus X_q\) is robust for \(v_1\), and by \(1\) in the opposite case.

We now show that adjacent vertices of the path carry the same label, this will eventually lead to a contradiction. Indeed, suppose \(X\) has label \(0\) and an adjacent vertex \(Y\) has label \(1\). Then every set in \(\lowerstar X\) has value zero for \(v_1\) and every set in \(\lowerstar Y\) has value one, so \(\lowerstar X\cap\lowerstar Y=\varnothing\); likewise every set in \(\lowerstar(M\setminus X)\) has value one and every set in \(\lowerstar(M\setminus Y)\) has value zero, so \(\lowerstar(M\setminus X)\cap\lowerstar(M\setminus Y)=\varnothing\). Both conditions in the definition of \(H_{\cK}\) fail, so \(X\) and \(Y\) are not adjacent. Hence, by contradiction, the label is fixed along the path. 

On the other hand, interchanging \(X\) and \(M\setminus X\) interchanges the two cases of the labeling, so \(\tau(X)\) carries the label opposite to that of \(X\); in particular the endpoints \(X_0=B\) and \(X_\ell=\tau(B)\) carry opposite labels. This contradiction implies that there must be at least one agreeable vertex on the path.\footnote{The logical structure of this argument deserves a remark, since the labelling above is available only under the supposition that no vertex of the path is agreeable. What the argument establishes is that \emph{at least one} vertex of the path is agreeable; it does not exhibit that vertex, and in particular it does not show that the chosen starting vertex \(B\) is agreeable. If \(B\) happens to be agreeable, then the path is not needed at all.}

Fix an agreeable vertex \(X\) of the path. Both \(X\) and \(M\setminus X\) are independent, because \(X\in\cF_{\cK}\). Offer the two bundles \(X\) and \(M\setminus X\) to agent \(2\), let her keep one that she weakly prefers, and give the other to agent \(1\). Agent \(2\) does not envy agent \(1\), and agent \(1\) is the source of no violation, whichever of the two bundles she receives, because \(X\) is agreeable for \(v_1\). The result is a complete feasible EF1 allocation for the profile \((v_1,v_2)\).

For necessity, suppose no connected component maps onto itself under \(\tau\). The components then occur in disjoint pairs \(K,\tau(K)\). Choose a labeling \(c\) of the components by zero and one such that \(c(\tau(K))=1-c(K)\). Define a Boolean valuation \(v\) as follows. If \(S\in\lowerstar B\) for some \(B\in\cF_{\cK}\), set \(v(S)=c(K_B)\), where \(K_B\) is the component containing \(B\); assign value zero to all remaining sets. This is well defined: if \(\lowerstar B\cap\lowerstar B' \neq \varnothing\), then either \(B=B'\) or \(B,B'\) are adjacent, and hence belong to the same component.
 
Fix \(B\in\cF_{\cK}\). Every set in \(\lowerstar B\) has value \(c(K_B)\), while every set in \(\lowerstar(M\setminus B)\) has value \(1-c(K_B)\). If \(c(K_B)=0\), then \(B\) is unsafe and \(M\setminus B\) robust; if \(c(K_B)=1\), the reverse holds. Consequently, under the identical profile \((v,v)\), every feasible allocation has one agent holding an unsafe bundle and the other holding a robust bundle. \Cref{lem:violation} shows that no feasible allocation is EF1, completing the proof.
\end{proof}
 
\subsection{Basis-pair reconfiguration and positive classes}
We now specialize \Cref{thm:matroid-characterization} to the case in which the ground set is exactly the union of two bases. We show that the graph \(H_{\cK}\) then records single exchanges of items between a pair of complementary bases, and we explain how the self-complementary-component condition relates to two longstanding conjectures on such exchanges; this yields feasible EF1 allocations for several classes of matroids.

Suppose that \(\cK\) has rank \(r\), that \(|M|=2r\), and that \(M\) can be partitioned into two independent sets; we refer to this as the \emph{base-partition case}. Each part of such a partition has size at most \(r\) and the two sizes add up to \(2r\), so both parts have size exactly \(r\) and are therefore bases. Hence \(\cF_{\cK}\) consists precisely of those bases whose complement is again a basis, and a vertex \(B\) of \(H_{\cK}\) encodes the ordered pair of complementary bases \((B,M\setminus B)\), with \(\tau\) interchanging the two entries of the pair.

Recall that a \emph{symmetric exchange}~\cite{white1980} is a procedure that transforms a pair of disjoint bases \((B,M\setminus B)\) into \((B-x+y,\,(M\setminus B)-y+x)\), where \(x\in B\) and \(y\in M\setminus B\) are such that both new sets are again bases.\footnote{For simplicity we denote $X \cup \{g\}$ and $X \setminus \{g\}$ using $X + g$ and $X - g$ respectively.} In the base-partition case, adjacency of two vertices in \(H_{\cK}\) exactly corresponds to the existence of a symmetric exchange between the two pairs corresponding to each vertex. Indeed, if \(B'=B-x+y\) arises from \(B\) by a symmetric exchange, then \(B\setminus\{x\}=B'\setminus\{y\}\) lies in \(\lowerstar B\cap\lowerstar B'\), so \(B\) and \(B'\) are adjacent. Conversely, let \(B\ne B'\) be adjacent, and suppose first that \(\lowerstar B\cap\lowerstar B'\ne\varnothing\). A common member of the two lower stars has size \(r\) or \(r-1\); it cannot have size \(r\), since that would force \(B=B'\), so it is of the form \(B\setminus\{x\}=B'\setminus\{y\}\) with \(x\in B\) and \(y\in B'\), whence \(B'=B-x+y\). Both complements are bases because \(B'\in\cF_{\cK}\), so this is a symmetric exchange. If instead \(\lowerstar(M\setminus B)\cap\lowerstar(M\setminus B')\ne\varnothing\), the same argument applied to the complements gives \(M\setminus B'=(M\setminus B)-y+x\), which is again a symmetric exchange. Thus \(H_{\cK}\) is the graph whose vertices are the pairs of complementary bases and whose edges join two pairs that differ by one symmetric exchange.

White~\cite{white1980} conjectured that any two pairs of bases with the same union can be transformed into one another by a sequence of symmetric exchanges, and Gabow~\cite{gabow1976} conjectured the stronger statement that such exchanges can be carried out in a serial fashion. These conjectures relate to the condition of \Cref{thm:matroid-characterization} as follows. In the base-partition case, the pairs \((B,M\setminus B)\) and \((M\setminus B,B)\) have the same union, namely \(M\); White's conjecture applied to them asserts that they are connected by symmetric exchanges, that is, by the previous paragraph, that \(B\) and \(\tau(B)\) lie in the same connected component of \(H_{\cK}\). This is precisely the statement that the component containing \(B\) is self-complementary, which is what \Cref{thm:matroid-characterization} requires. Note that only one vertex \(B\) and its complement are needed, so the condition of \Cref{thm:matroid-characterization} is weaker than the conjectures.

The conjectures are known to hold for several classes of matroids, which the previous paragraph converts into fairness guarantees. B\'erczi and Schwarcz prove the required exchange property for split matroids, a class that includes sparse paving matroids~\cite{berczi-schwarcz}, and B\'erczi, M\'atrav\"olgyi, and Schwarcz prove it for regular matroids~\cite{berczi-matravolgyi-schwarcz}. For a strongly base-orderable matroid the property is immediate: there is a bijection \(\phi:B\to M\setminus B\) such that exchanging the elements of any subset of the pairs \(\{x,\phi(x)\}\) leaves both sets bases, so performing these exchanges one pair at a time traces a path in \(H_{\cK}\) from \(B\) to \(\tau(B)\). Combining these facts with \Cref{thm:matroid-characterization} gives the following.

\begin{corollary}
\label{cor:matroid-classes}
In a two-agent base-partition instance, arbitrary Boolean valuations admit a feasible EF1 allocation whenever the matroid is strongly base-orderable, split, sparse paving, or regular.
\end{corollary}
 
The exact condition in \Cref{thm:matroid-characterization} is weaker than global basis-pair connectivity: it requires only one self-complementary component. Consequently, the theorem identifies the precise reconfiguration property needed by arbitrary Boolean valuations without assuming a stronger open conjecture. The following corollary follows by padding a feasible instance to the base-partition case and using the fact that uniform matroids and partition matroids are strongly base-orderable.
 
\begin{corollary}
\label{cor:uniform-partition}
For two agents, arbitrary Boolean valuations admit a complete feasible EF1 allocation under every uniform-matroid or partition-matroid constraint for which a complete feasible allocation exists.
\end{corollary}
 
\begin{proof}
For a uniform matroid of rank \(r\), the existence of a complete feasible allocation implies \(|M|\le2r\). Add \(2r-|M|\) dummy items; the enlarged rank-\(r\) uniform matroid has a ground set of size \(2r\) and can therefore be partitioned into two bases. For a partition matroid with blocks \(P_q\) and capacities \(r_q\), the existence of a complete feasible allocation implies \(|P_q|\le2r_q\) for every block. Add \(2r_q-|P_q|\) dummy items to block \(P_q\). Each enlarged block then has size \(2r_q\), so splitting every block into two sets of size \(r_q\) partitions the enlarged ground set into two bases. Thus, in either case, the enlarged instance is a base-partition instance of a strongly base-orderable matroid.

Extend each valuation by setting \(\widetilde v_i(S)=v_i(S\cap M)\), and apply \Cref{cor:matroid-classes} to obtain a feasible EF1 allocation \(\widetilde A\) of the enlarged instance. Set \(A_i=\widetilde A_i\cap M\) for each agent \(i\). The resulting allocation \(A\) is complete and feasible for the original instance. Fix an ordered pair of agents \(i,j\). If \(\widetilde v_i(\widetilde A_i)\ge\widetilde v_i(\widetilde A_j)\), then \(v_i(A_i)\ge v_i(A_j)\). Otherwise, let \(e\in\widetilde A_i\cup\widetilde A_j\) be an EF1 witness. If \(e\in M\), then the same item witnesses EF1 for \(A\). If \(e\) is a dummy item, then \(\widetilde v_i(\widetilde A_i\setminus\{e\})\ge\widetilde v_i(\widetilde A_j\setminus\{e\})\) reduces to \(v_i(A_i)\ge v_i(A_j)\), so no witness is needed. Hence \(A\) is EF1.
\end{proof}
 
The question for arbitrary Boolean valuations under matroid constraints is therefore already tied, in the two-agent case, to a longstanding basis-pair reconfiguration problem. 

\subsection{EF1 and Pareto optimality can conflict under a partition matroid}

We show that EF1 and PO are incompatible under partition-matroid constraints. Pareto optimality in this section is relative to the family of complete feasible allocations.

\begin{proposition}
\label{prop:matroid-po}
There is a two-agent instance with identical normalized Boolean valuations and a partition-matroid constraint in which feasible EF1 allocations and feasible Pareto-optimal allocations both exist, but no feasible allocation satisfies both properties.
\end{proposition}
 
\begin{proof}
Let \(M=\{a,b,c,d\}\). The partition matroid has two capacity-one blocks, \(\{a,c\}\) and \(\{b,d\}\). Let both agents have the normalized valuation
\[
v(S)=1\quad\Longleftrightarrow\quad S\in\bigl\{\{c\},\{d\},\{c,d\}\bigr\}.
\]
There are four ordered complete feasible allocations. The allocations with bundle pair \(\{a,b\},\{c,d\}\), in either orientation, have utility vectors \((0,1)\) and \((1,0)\). They are the only Pareto-optimal allocations: keeping the agent who receives value one at value one forces her to retain \(\{c,d\}\). Yet \(\{a,b\}\) is unsafe and \(\{c,d\}\) is robust, so each such allocation violates EF1 by \Cref{lem:violation}.
 
The remaining two allocations have bundle pairs \(\{a,d\},\{b,c\}\), in either orientation. Both agents receive value zero, so these allocations are envy-free and hence EF1. Each is Pareto dominated by orienting \(\{c,d\}\) toward one agent and \(\{a,b\}\) toward the other. Thus no feasible allocation is both EF1 and Pareto optimal.
\end{proof}

\section{Discussion}
\label{sec:discussion}
 
As our results highlight, the value of the empty bundle is more than a normalization convention: it dictates which guarantees are attainable. An item may raise one bundle from zero to one and lower another from one to zero, so calling it a good or a chore is generally meaningless; somewhat surprisingly, the value agents assign to \(\varnothing\) dictates which fairness guarantees we can achieve. %The sharp transition at \(|Z|=1\) makes the point especially vivid. 

Our mechanism-design results leave one question open. On \(\cVz^n\), \Cref{thm:score-mechanism} gives a weakly group-strategyproof and Pareto-optimal mechanism whose output is \(\EFXzm\). At the opposite endpoint \(Z=\varnothing\), \Cref{prop:po-impossibility} shows that, for every \(n,m\ge2\), there exists a profile at which no allocation is simultaneously Pareto optimal and EF1; hence no mechanism can guarantee both properties throughout \(\cVo^n\). However, existence of a strategyproof mechanism that always outputs EF1 allocations on \(\cV^n\) remains open.

%Can a strategyproof rule guarantee EF1 on \(\cVo^n\) for arbitrary $n$ and \(m\)? On which Boolean profiles can strategyproofness, EF1, and Pareto optimality be achieved simultaneously?

For two agents we obtain a complete answer for the fairness achievable under matroid feasibility constraints: a feasible EF1 allocation is guaranteed for every Boolean profile exactly when the exchange graph \(H_{\cK}\) has a self-complementary component (\Cref{thm:matroid-characterization}). This condition is weaker than the exchange property conjectured by White and by Gabow, and it already holds for several classes of matroids, which is how \Cref{cor:matroid-classes} is obtained. Proving this weaker property for arbitrary matroids is a natural question. Additionally, extending the characterization to more than two agents would be interesting.

%Extending the characterization to more agents seems to require a different object, since a graph records only exchanges between two bundles whereas \(n\) bundles may be permuted in many ways at once; agent-specific matroids are further out of reach, as there is then no single exchange graph shared across agents.

%Which matroids guarantee a feasible EF1 allocation for every Boolean profile and every number of agents? In the two-agent case, must the basis-pair exchange graph contain a self-complementary component?

%The natural next domain is preferences with three indifference classes rather than two. There a single deletion can move a bundle by more than one level, so the distinction between unsafe and robust bundles that drives all of our arguments no longer captures the effect of removing an item, and it is not clear which of our results survive.

%%%%%%%%%%%%%%%%%%%%%%%%%%%%%%%%%%%%%%%%%%%%%%%%%%%%%%%%%%%%%%%%%%%%%%%%

\section*{Acknowledgements}

The authors were supported by an NSERC Discovery Grant and an NSERC-CSE Research Communities Grant. Researchers funded through the NSERC-CSE Research Communities Grants do not represent the Communications Security Establishment Canada or the Government of Canada. Any research, opinions or positions they produce as part of this initiative do not represent the official views of the Government of Canada.

\section*{AI Disclosure}

Many of the mathematical proofs and counterexamples were derived by OpenAI Codex (GPT-5.6-Sol at Max effort) based on explicit questions, research directions, literature connections, proof and search strategies, and inspirations supplied by the authors. The authors have verified all mathematical details, presented simplified arguments and expositions, often with the aid of GPT-5.6-Sol and Claude Opus 5. The authors retain full responsibility for all the content.

\bibliographystyle{plainnat}
\bibliography{sample}

@article{babaioff2021dichotomous,
  author    = {Babaioff, Moshe and Ezra, Tomer and Feige, Uriel},
  title     = {Fair and Truthful Mechanisms for Dichotomous Valuations},
  journal   = {Proceedings of the AAAI Conference on Artificial Intelligence},
  volume    = {35},
  number    = {6},
  pages     = {5119--5126},
  year      = {2021},
  doi       = {10.1609/aaai.v35i6.16647}
}

@inproceedings{halpern2020binary,
  author    = {Halpern, Daniel and Procaccia, Ariel D. and Psomas, Alexandros and Shah, Nisarg},
  title     = {Fair Division with Binary Valuations: One Rule to Rule Them All},
  booktitle = {Web and Internet Economics --- 16th International Conference (WINE)},
  series    = {Lecture Notes in Computer Science},
  volume    = {12495},
  pages     = {370--383},
  publisher = {Springer},
  year      = {2020},
  doi       = {10.1007/978-3-030-64946-3_26}
}

@inproceedings{barman-verma-mms,
  author    = {Barman, Siddharth and Verma, Paritosh},
  title     = {Existence and Computation of Maximin Fair Allocations under Matroid-Rank Valuations},
  booktitle = {Proceedings of the 20th International Conference on Autonomous Agents and Multiagent Systems (AAMAS)},
  pages     = {169--177},
  year      = {2021}
}

@inproceedings{barman-verma-truthful,
  author    = {Barman, Siddharth and Verma, Paritosh},
  title     = {Truthful and Fair Mechanisms for Matroid-Rank Valuations},
  booktitle = {Proceedings of the 36th AAAI Conference on Artificial Intelligence (AAAI)},
  pages     = {4801--4808},
  year      = {2022}
}

@inproceedings{viswanathan-zick-yankee,
  author    = {Viswanathan, Vignesh and Zick, Yair},
  title     = {Yankee Swap: A Fast and Simple Fair Allocation Mechanism for Matroid Rank Valuations},
  booktitle = {Proceedings of the 22nd International Conference on Autonomous Agents and Multiagent Systems (AAMAS)},
  pages     = {179--187},
  year      = {2023}
}

@inproceedings{viswanathan-zick-framework,
  author    = {Viswanathan, Vignesh and Zick, Yair},
  title     = {A General Framework for Fair Allocation under Matroid Rank Valuations},
  booktitle = {Proceedings of the 24th ACM Conference on Economics and Computation (EC)},
  pages     = {1129--1152},
  year      = {2023},
  doi       = {10.1145/3580507.3597675}
}

@article{benabbou2021matroid,
  author  = {Benabbou, Nawal and Chakraborty, Mithun and Igarashi, Ayumi and Zick, Yair},
  title   = {Finding Fair and Efficient Allocations for Matroid Rank Valuations},
  journal = {ACM Transactions on Economics and Computation},
  volume  = {9},
  number  = {4},
  pages   = {21:1--21:41},
  year    = {2021}
}

@article{aziz2022,
  author  = {Haris Aziz and Ioannis Caragiannis and Ayumi Igarashi and Toby Walsh},
  title   = {Fair Allocation of Indivisible Goods and Chores},
  journal = {Autonomous Agents and Multi-Agent Systems},
  volume  = {36},
  number  = {1},
  pages   = {3},
  year    = {2022},
  doi     = {10.1007/s10458-021-09532-8}
}

@article{berczi2024,
  author  = {Krist{\'o}f B{\'e}rczi and Erika R. B{\'e}rczi-Kov{\'a}cs and Endre Boros and Fekadu Tolessa Gedefa and Naoyuki Kamiyama and Telikepalli Kavitha and Yusuke Kobayashi and Kazuhisa Makino},
  title   = {Envy-free Relaxations for Goods, Chores, and Mixed Items},
  journal = {Theoretical Computer Science},
  volume  = {1002},
  pages   = {114596},
  year    = {2024},
  doi     = {10.1016/j.tcs.2024.114596}
}

@article{berczi-matravolgyi-schwarcz,
  author  = {Krist{\'o}f B{\'e}rczi and Bence M{\'a}trav{\"o}lgyi and Tam{\'a}s Schwarcz},
  title   = {Reconfiguration of Basis Pairs in Regular Matroids},
  journal = {Journal of Combinatorial Theory, Series B},
  volume  = {177},
  pages   = {105--142},
  year    = {2026},
  doi     = {10.1016/j.jctb.2025.10.009},
  note    = {Preliminary version in \emph{Proceedings of the 56th Annual ACM Symposium on Theory of Computing (STOC)}, pages 1653--1664, 2024}
}

@article{berczi-schwarcz,
  author  = {Krist{\'o}f B{\'e}rczi and Tam{\'a}s Schwarcz},
  title   = {Exchange Distance of Basis Pairs in Split Matroids},
  journal = {SIAM Journal on Discrete Mathematics},
  volume  = {38},
  number  = {1},
  pages   = {132--147},
  year    = {2024},
  doi     = {10.1137/23M1565115}
}

@inproceedings{bhaskar2021,
  author    = {Umang Bhaskar and A. R. Sricharan and Rohit Vaish},
  title     = {On Approximate Envy-Freeness for Indivisible Chores and Mixed Resources},
  booktitle = {Approximation, Randomization, and Combinatorial Optimization. Algorithms and Techniques (APPROX/RANDOM)},
  series    = {LIPIcs},
  volume    = {207},
  pages     = {1:1--1:23},
  publisher = {Schloss Dagstuhl -- Leibniz-Zentrum f{\"u}r Informatik},
  year      = {2021},
  doi       = {10.4230/LIPIcs.APPROX/RANDOM.2021.1}
}

@inproceedings{bhaskar2025,
  author    = {Umang Bhaskar and Gunjan Kumar and Yeshwant Pandit and Rakshitha},
  title     = {Towards Envy-Freeness Relaxations for General Nonmonotone Valuations},
  booktitle = {Proceedings of the 24th International Conference on Autonomous Agents and Multiagent Systems (AAMAS)},
  pages     = {298--306},
  publisher = {International Foundation for Autonomous Agents and Multiagent Systems},
  year      = {2025},
  note      = {Extended version available as arXiv:2411.19881}
}

@article{gabow1976,
  author  = {Harold N. Gabow},
  title   = {Decomposing Symmetric Exchanges in Matroid Bases},
  journal = {Mathematical Programming},
  volume  = {10},
  number  = {1},
  pages   = {271--276},
  year    = {1976},
  doi     = {10.1007/BF01580677}
}

@inproceedings{lipton2004,
  author    = {Richard J. Lipton and Evangelos Markakis and Elchanan Mossel and Amin Saberi},
  title     = {On Approximately Fair Allocations of Indivisible Goods},
  booktitle = {Proceedings of the 5th ACM Conference on Electronic Commerce (EC)},
  pages     = {125--131},
  publisher = {ACM},
  year      = {2004},
  doi       = {10.1145/988772.988792}
}

@article{white1980,
  author  = {Neil L. White},
  title   = {A Unique Exchange Property for Bases},
  journal = {Linear Algebra and its Applications},
  volume  = {31},
  pages   = {81--91},
  year    = {1980},
  doi     = {10.1016/0024-3795(80)90209-8}
}

@inproceedings{barman2026fair,
  author    = {Siddharth Barman and Paritosh Verma},
  title     = {Fair Division Beyond Monotone Valuations with Applications to Equitable Graph Partitioning},
  booktitle = {Proceedings of the 2026 Annual ACM--SIAM Symposium on Discrete Algorithms (SODA)},
  pages     = {6613--6641},
  year      = {2026},
  doi       = {10.1137/1.9781611978971.236}
}

@article{bilo2026approximately,
  author    = {Bil{\`o}, Vittorio and Loebl, Martin and Vinci, Cosimo},
  title     = {Approximately Envy-Free and Equitable Allocations of Indivisible Items for Non-monotone Valuations},
  journal   = {Proceedings of the AAAI Conference on Artificial Intelligence},
  volume    = {40},
  number    = {20},
  pages     = {16700--16708},
  year      = {2026},
  doi       = {10.1609/aaai.v40i20.38712}
}

@article{edward1999rental,
  author  = {Francis Edward Su},
  title   = {Rental Harmony: Sperner's Lemma in Fair Division},
  journal = {The American Mathematical Monthly},
  volume  = {106},
  number  = {10},
  pages   = {930--942},
  year    = {1999},
  doi     = {10.1080/00029890.1999.12005142}
}

%%%%%%%%%%%%%%%%%%%%%%%%%%%%%%%%%%%%%%%%%%%%%%%%%%%%%%%%%%%%%%%%%%%%%%%%

\onecolumn 
\appendix
\crefalias{section}{appendix}

\section*{\centering Appendix}

\section{Valuation classes}
\label{sec:valuation-classes}
 
Several valuation classes appear in the discussion of related work; we collect their definitions
here. Throughout this section, \(v:2^M\to\bR\) denotes the valuation of a single agent,
normalized so that \(v(\varnothing)=0\). For a bundle \(S\subseteq M\) and an item \(e\notin S\), the
\emph{marginal value} of \(e\) with respect to \(S\) is
\[
  v(e\mid S)\;=\;v(S\cup\{e\})-v(S),
\]
and \(v\) is \emph{submodular} if marginal values never increase as the bundle grows, that is, if
\(v(e\mid S)\ge v(e\mid T)\) whenever \(S\subseteq T\subseteq M\setminus\{e\}\).
 
We first record the classes in which the value of a bundle is the sum of the values of its items.
 
\begin{itemize}
\item \emph{Additive}: \(v(S)=\sum_{e\in S}v(\{e\})\) for every bundle \(S\), where \(v(\{e\})\ge0\)
for every item \(e\).
 
\item \emph{Binary additive}: additive, with the further requirement that \(v(\{e\})\in\{0,1\}\) for
every item \(e\). Each item is thus either wanted or unwanted, and the value of a bundle counts the
wanted items that it contains.
 
\item \emph{Additive mixed manna}: \(v(S)=\sum_{e\in S}v(\{e\})\) for every bundle \(S\), where
\(v(\{e\})\) is now an arbitrary real number. An item \(e\) is a \emph{good} if \(v(\{e\})>0\), a
\emph{chore} if \(v(\{e\})<0\), and irrelevant to the agent if \(v(\{e\})=0\).
\end{itemize}
 
The next two classes keep the requirement that every marginal value be zero or one, but drop
additivity.
 
\begin{itemize}
\item \emph{Dichotomous}: \(v(e\mid S)\in\{0,1\}\) for every bundle \(S\) and every item \(e\notin S\). Such a valuation is monotone nondecreasing, and its range is contained in \(\{0,1,\ldots,m\}\).
 
\item \emph{Matroid rank}, also called \emph{binary submodular}: dichotomous and submodular.
Equivalently, \(v\) is the rank function of a matroid on \(M\), so that \(v(S)\) is the size of a
largest independent subset of \(S\). Every binary additive valuation is of this form.
\end{itemize}
 
The remaining classes are defined by the signs of the marginal values alone, and impose no
structure beyond them.
 
\begin{itemize}
\item \emph{Monotone nondecreasing}: \(v(S)\le v(T)\) whenever \(S\subseteq T\); equivalently, \(v(e\mid S)\ge0\)
for every bundle \(S\) and every item \(e\notin S\), so that every item is a good.
 
\item \emph{Doubly monotone}: the item set splits into two parts, \(M^{+}\) and \(M^{-}\), such that
\(v(e\mid S)\ge0\) for every \(e\in M^{+}\) and \(v(e\mid S)\le0\) for every \(e\in M^{-}\), for
every bundle \(S\). Every item is therefore a good or a chore for the agent, but the split may
differ from agent to agent, and the value of a bundle need not be the sum of the values of its
items.
 
\item \emph{Nonnegative}: \(v(S)\ge0\) for every bundle \(S\), with no monotonicity requirement.
 
\item \emph{Arbitrary set valuations}: no requirement beyond \(v\) being a function on \(2^M\).
\end{itemize}
 
The Boolean valuations studied in this paper are nonnegative, but they are subject to no other
requirement, and in particular they are incomparable with the classes based on binary marginals: a
Boolean valuation has range \(\{0,1\}\) and need not be monotone nondecreasing, whereas a dichotomous valuation is monotone nondecreasing and may take any value in \(\{0,1,\ldots,m\}\). The two classes intersect precisely in the monotone nondecreasing valuations of \(\cVz\).
 
Finally, a profile \((v_1,\ldots,v_n)\) is \emph{identical} if all agents share the same valuation. We use \emph{heterogeneous} to mean \emph{not necessarily identical}: no restriction is imposed on how the agents' valuations relate to one another.

% \section{Missing proofs from~\Cref{sec:existence}}
% \label{app:existence}

% \efoneExistence*

% \section{Missing proofs from~\Cref{sec:onesided}}
% \label{app:onesided}
%\propEFxEFOneequivalence*

% \goodsOPT*

% \NoEFoneplus*

% \choreSide*

% \POEFZOne*
 
% \propNegativeEFOneMinus*
% \begin{proof}
% Let $z_1,z_2\in Z$ be two agents that value the empty bundle at zero and take $m=1$ such that $v_{z_1}=v_{z_2}$ equal to $\mathbf 1[S=\{g\}]$; pick the valuation functions of the other agents arbitrarily. Under any allocation the single item is
% held by one agent, so at least one of $z_1,z_2$ holds $\varnothing$. This agent's envy towards the agent having the single item cannot be eliminated via removal of an item from its own bundle (which is empty). Hence no allocation is $\EFoneminus$.
% \end{proof}

% \POEFoneCounterexample*

% \section{Missing proofs from~\Cref{sec:efx}}
% \label{app:efx}

% \EFXPZCharacterization*

% \efxZminus*

\section{Strong group-strategyproofness counter example from~\Cref{sec:mechanisms}}
\label{app:mechanisms}

% \mechanism*

\begin{example}[The score mechanism is not strongly group-strategyproof]\label{example:notSGSP}
Let \(N=\{1,2,3\}\), \(M=\{a,b,c,d\}\), \(B_2=\{a,b\}\), and \(B_3=\{b,c\}\). Let \(v_1(S)=0\) for every \(S\subseteq M\), and, for each \(i\in\{2,3\}\), let \(v_i(S)=1\) if and only if \(S=B_i\). At the truthful profile, the maximum welfare is one because \(B_2\cap B_3\ne\varnothing\), and every welfare maximizer has secondary score two. Let \(X=\mathcal M_\rhd(v)\), let \(w\in\{2,3\}\) be the unique agent with \(X_w=B_w\), let \(\ell\) be the other agent in \(\{2,3\}\), and put \(T=M\setminus B_\ell\). Consider the coalition \(K=\{1,\ell\}\). Agent \(1\) reports \(\widehat v_1(S)=1\) if and only if \(S=T\), agent \(\ell\) reports \(\widehat v_\ell(S)=1\) if and only if \(S\in\{B_\ell,M\}\), and agent \(w\) reports truthfully. The bundles \(T\) and \(B_\ell\) partition \(M\), while \(B_w\) intersects both of them. Therefore the unique reported-welfare-two allocation gives \(T\) to agent \(1\), \(B_\ell\) to agent \(\ell\), and \(\varnothing\) to agent \(w\); no allocation has reported welfare three. The mechanism selects this allocation regardless of its secondary score or the fixed tie-breaking rule. Under the true valuations, agent \(1\)'s value remains zero, while agent \(\ell\)'s value increases from zero to one. Thus every coalition member is weakly better off and one is strictly better off, so \(\mathcal M_\rhd\) is not strongly group-strategyproof.
\end{example}

% \section{Missing proofs from~\Cref{sec:bobw}}
% \label{app:bobw}

\end{document}